\documentclass[11pt]{article}

\usepackage{mathpazo}
\usepackage{amsfonts}
\usepackage{amsmath}
\usepackage{graphicx}
\usepackage{latexsym}
\usepackage{bm}

\usepackage[margin=1.05in]{geometry}
  
\usepackage[T1]{fontenc}
\usepackage{times}
\usepackage{color,graphicx}
\usepackage{array}
\usepackage{enumerate}
\usepackage{amsmath}
\usepackage{amssymb}
\usepackage{amsthm}
\usepackage{pgfplots}
\usepackage{pgf}
\usepackage{tikz}
\usetikzlibrary{patterns}
\usetikzlibrary{shapes.geometric} 
\usepgfplotslibrary{patchplots} 
\usetikzlibrary{pgfplots.patchplots} 
\pgfplotsset{width=9cm,compat=1.5.1}
\usepackage{amsthm}
\newtheorem{conjecture}{Conjecture}
\newtheorem{definition}{Definition}
\newtheorem{lemma}{Lemma}
\newtheorem{theorem}{Theorem}
\newtheorem{corollary}{Corollary}

\newtheorem{example}{Example}

\usepackage{xcolor}
\usepackage{makeidx}
\usepackage[colorlinks=true,linkcolor=blue,anchorcolor=blue,citecolor=red,urlcolor=magenta]{hyperref}

\newcommand{\tr}{\operatorname{Tr}}

\newcommand{\Sym}[1]{\mathbb{C}^{#1} \vee \mathbb{C}^{#1}}
\newcommand{\RSym}[1]{\mathbb{R}^{#1} \vee \mathbb{R}^{#1}}

\newcommand\Tr{\mathop{\rm Tr}\nolimits}

\newcommand{\defeq}{\stackrel{\smash{\textnormal{\tiny def}}}{=}}

\def\R{\mathbb{R}}

\renewcommand{\Re}[1]{\mathrm{Re}\left(#1 \right)}
\renewcommand{\Im}[1]{\mathrm{Im}\left(#1 \right)}

\begin{document}
\title{Spectral Properties of Symmetric Quantum States \\ and Symmetric Entanglement Witnesses}

\author{
	Gabriel Champagne\footnote{Department of Mathematics \& Computer Science, Mount Allison University, Sackville, NB, Canada E4L 1E4}, 
	\quad Nathaniel Johnston$^{*,}$\footnote{Department of Mathematics \& Statistics, University of Guelph, Guelph, ON, Canada N1G 2W1}\!\!\!, \ \ \quad Mitchell MacDonald$^{*}$\!\!, \quad and \quad 
	Logan Pipes$^{*}$
}

\maketitle

\begin{abstract}
	We introduce and explore two questions concerning spectra of operators that are of interest in the theory of entanglement in symmetric (i.e., bosonic) quantum systems. First, we investigate the inverse eigenvalue problem for symmetric entanglement witnesses---that is, we investigate what their possible spectra are. Second, we investigate the problem of characterizing which separable symmetric quantum states remain separable after conjugation by an arbitrary unitary acting on symmetric space---that is, which states are separable in every orthonormal symmetric basis. Both of these questions have been investigated thoroughly in the non-symmetric setting, and we contrast the answers that we find with their non-symmetric counterparts.\\
	
	
\end{abstract}

\section{Introduction}

In quantum information theory, there is considerable interest in the problem of determining whether a given quantum state is separable or entangled \cite{GT09,HHH09}. This problem has received the most attention in systems of distinguishable quantum particles, where separability and entanglement are modeled by the tensor (i.e., Kronecker) product (see any number of textbooks on quantum information theory, like \cite{NC00,Wat18}). However, this separability problem has also been explored in systems of indistinguishable particles \cite{ESBL02,TG10,Yu16}, and many of the same techniques are applicable in that setting.

One standard method of showing that a quantum state is entangled is to use an entanglement witness---a Hermitian operator with a non-negative expectation on all separable states (and thus a negative expectation with a given state certifies that it must be entangled). Entanglement witnesses are dual to separable states in a natural sense \cite{SSZ09}, so answering questions about one of these sets typically leads to answers for questions about the other set.

Numerous results are known that give bounds on the possible eigenvalues of (non-symmetric) entanglement witnesses---see \cite{JP18,Sar08} and the references therein. However, to our knowledge the only known result about the spectrum of \emph{symmetric} entanglement witnesses is that they can have at most $\binom{d}{2} = d(d-1)/2$ negative eigenvalues \cite[Corollary~15]{LJ20}, and this bound is tight (i.e., for all $d$ there is a symmetric entanglement witness with this many negative eigenvalues). Equivalently, since they have $\binom{d+1}{2} = d(d+1)/2$ eigenvalues in total, at least $\binom{d+1}{2} - \binom{d}{2} = d$ of their eigenvalues must be non-negative. We explore the problem of bounding the spectra of symmetric entanglement witnesses in more depth in the coming sections. We almost completely characterize the possible spectra of symmetric entanglement witnesses in the $d = 2$ case in Theorem~\ref{thm:2x2_eigs}, and we derive other spectral facts and bounds throughout the rest of Section~\ref{sec:witness_spectra}.

A rather complementary question asks which quantum states have the property that they can be determined to be separable solely based on their eigenvalues. Such states are said to be \emph{absolutely separable} \cite{KZ00}, and they have been completely characterized in small dimensions \cite{Joh13} and have been bounded in all dimensions \cite{Hil07}. We introduce the analogous problem for \emph{symmetric} separable states---we ask which spectra are such that they guarantee separability of a symmetric quantum state. We completely solve this problem in the $d = 2$ case (Theorem~\ref{thm:asymppt_from_eigenvalues_2dim}) and provide bounds in all dimensions (Theorem~\ref{thm:asymppt_from_eigenvalues} and~Corollary~\ref{cor:asymppt_from_eigenvalues_3dim}).

This paper is organized as follows: In Section~\ref{sec:math_prelims}, we introduce our notation and terminology, as well as the relevant mathematical background for the problems that we consider; in Section~\ref{sec:decomp_sew}, we introduce some simple lemmas that we will need throughout the rest of the paper; in Section~\ref{sec:abs_sep}, we present and prove our main results about the absolute separability problem for symmetric states; in Section~\ref{sec:witness_spectra}, we present and prove our main results about the spectra of symmetric entanglement witnesses; and in Section~\ref{sec:conclusions} we close with some open questions and potential directions for future research.

\section{Mathematical Preliminaries}\label{sec:math_prelims}

We use $\mathcal{L}(\mathcal{V})$ to denote the set of linear maps acting on a vector space $\mathcal{V}$. All vector spaces that we consider will be finite-dimensional, so we represent members of $\mathcal{L}(\mathcal{V})$ as $\mathrm{dim}(\mathcal{V}) \times \mathrm{dim}(\mathcal{V})$ matrices when it is convenient to do so. The vector space of (ordered) tuples of $d$ complex numbers is denoted by $\mathbb{C}^d$, and the standard basis vectors in $\mathbb{C}^d$ are denoted by $\mathbf{e_j} = (0,0,\ldots,0,1,0,\ldots,0)$ for $1 \leq j \leq d$, where the single ``$1$'' occurs in the $j$-th entry.

In quantum information theory, a \emph{quantum state} is a positive semidefinite (PSD) Hermitian matrix $\rho \in \mathcal{L}(\mathcal{V})$ with $\tr(\rho) = 1$, where $\mathcal{V}$ is a vector space over the field $\mathbb{C}$ of complex numbers. In many sources, the explicit choice of $\mathcal{V} = \mathbb{C}^d$ is made, but for us it will sometimes be convenient to make other choices for $\mathcal{V}$ instead. Whenever we use lowercase Greek letters like $\rho$ or $\sigma$, we are implicitly assuming that they represent quantum states and thus have these positive semidefiniteness and trace properties.

Throughout this work, we frequently consider the eigenvalues of Hermitian operators. We always sort them in non-increasing order and denote them by either $\lambda_1 \geq \lambda_2 \geq \lambda_3 \geq \cdots$ (typically for eigenvalues of quantum states) or $\mu_1 \geq \mu_2 \geq \mu_3 \geq \cdots$ (typically for eigenvalues of entanglement witnesses). We similarly denote singular values by $\sigma_1 \geq \sigma_2 \geq \sigma_3 \geq \cdots$.

\subsection{Separability and Entanglement}

A quantum state $\rho \in \mathcal{L}(\mathbb{C}^{d_1} \otimes \mathbb{C}^{d_2})$ is called \emph{separable} \cite{Wer89} if there exist vectors $\{\mathbf{v_j}\} \subset \mathbb{C}^{d_1}$ and $\{\mathbf{w_j}\} \subset \mathbb{C}^{d_2}$ such that
\begin{align}\label{eq:sep_decomp}
	\rho = \sum_{j} \mathbf{v_j}\mathbf{v_j}^* \otimes \mathbf{w_j}\mathbf{w_j}^*,
\end{align}
where ``$\otimes$'' is the tensor (i.e., Kronecker) product. If $\rho$ is not separable then it is called \emph{entangled}.

Determining whether or not a quantum state is separable is a hard problem \cite{Gur03,Gha10}, but one useful method of partially solving this problem is to use the \emph{partial transpose} operation that is defined on matrices $A = \sum_j B_j \otimes C_j \in \mathcal{L}(\mathbb{C}^{d_1} \otimes \mathbb{C}^{d_2})$ by
\[
	A^\Gamma := \sum_j B_j \otimes C_j^T.
\]
Indeed, if $\rho$ is separable then $\rho^\Gamma$ is positive semidefinite \cite{HHH96,Per96}. If a quantum state $\rho$ is such that $\rho^\Gamma$ is positive semidefinite then we say that it has \emph{positive partial transpose (PPT)}, and we just stated (in slightly different language) that the set of separable states is a subset of the set of PPT states. Remarkably, the converse also holds (i.e., every PPT state is separable) when $d_1 d_2 \leq 6$ \cite{Sto63,Wor76,HHH96}, but it fails in larger dimensions.

One method of demonstrating entanglement in higher dimensions is to use an \emph{entanglement witness}---a Hermitian matrix $W \in \mathcal{L}(\mathbb{C}^{d_1} \otimes \mathbb{C}^{d_2})$ that has property that
\[
    (\mathbf{v} \otimes \mathbf{w})^*W(\mathbf{v} \otimes \mathbf{w}) \geq 0 \quad \text{for all} \quad \mathbf{v} \in \mathbb{C}^{d_1} \ \ \text{and} \ \ \mathbf{w} \in \mathbb{C}^{d_2}.\footnote{Many sources also require that $W$ not be positive semidefinite in order to qualify as an entanglement witness. We do not impose that restriction here, because (a) it is trivial to impose or not impose as needed, and (b) spectral conditions like the ones we will discuss in this paper are much easier to state without the not-PSD restriction.}
\]
Constructing entanglement witnesses is typically difficult, but once we have one, simply checking that a particular quantum state $\rho$ is such that $\tr(W\rho) < 0$ shows that it must be entangled.

One family of entanglement witnesses that are easy to construct are those of the form $W = X^\Gamma + Y$, where $X, Y \in \mathcal{L}(\mathbb{C}^{d_1} \otimes \mathbb{C}^{d_2})$ are Hermitian positive semidefinite. Matrices of this form are called \emph{decomposable} entanglement witnesses, and they really are entanglement witnesses since
\begin{align*}
    (\mathbf{v} \otimes \mathbf{w})^*W(\mathbf{v} \otimes \mathbf{w}) & = (\mathbf{v} \otimes \mathbf{w})^*X^\Gamma(\mathbf{v} \otimes \mathbf{w}) + (\mathbf{v} \otimes \mathbf{w})^*Y(\mathbf{v} \otimes \mathbf{w}) \\
    & = \tr\big(X(\mathbf{v}\mathbf{v}^* \otimes \overline{\mathbf{w}\mathbf{w}^*})\big) + (\mathbf{v} \otimes \mathbf{w})^*Y(\mathbf{v} \otimes \mathbf{w}) \geq 0,
\end{align*}
where the final inequality follows from non-negativity of all matrices involved. However, the only quantum states that decomposable entanglement witnesses can show are entangled are those that are not PPT.

\subsection{The Symmetric Subspace and Symmetric States}

The \emph{symmetric subspace} of $\mathbb{C}^d \otimes \mathbb{C}^d$ is the set
\[
    \mathbb{C}^d \vee \mathbb{C}^d \defeq \mathrm{span}\big\{ \mathbf{v} \otimes \mathbf{v} : \mathbf{v} \in \mathbb{C}^d \big\},
\]
which has dimension $\binom{d+1}{2} = d(d+1)/2$ \cite[Theorem~3.1.10]{JohALA}. The orthogonal projection onto the symmetric subspace ($P_{\vee}$) is completely specified by the fact that
\[
    P_{\vee}(\mathbf{e_i} \otimes \mathbf{e_j}) = \frac{1}{2}(\mathbf{e_i} \otimes \mathbf{e_j} + \mathbf{e_j} \otimes \mathbf{e_i}) \quad \text{for all} \quad 1 \leq i, j \leq d,
\]
where $\mathbf{e_j}$ is the $j$-th standard basis vector of $\mathbb{C}^d$ (i.e., $\mathbf{e_j}$ is the vector with $1$ in its $j$-th entry and $0$ in elsewhere).

A quantum state $\rho$ is called \emph{symmetric} if $\rho \in \mathcal{L}(\mathbb{C}^d \vee \mathbb{C}^d)$. Alternatively, we say that $\rho \in \mathcal{L}(\mathbb{C}^d \otimes \mathbb{C}^d)$ is symmetric if $P_{\vee}\rho P_{\vee} = \rho$, and we switch freely between these two viewpoints as there is an obvious isomorphism between them. For our purposes, the only salient difference between these two perspectives is that a symmetric quantum state in $\mathcal{L}(\mathbb{C}^d \vee \mathbb{C}^d)$ has exactly $d(d+1)/2$ eigenvalues, whereas the corresponding symmetric quantum state in $\mathcal{L}(\mathbb{C}^d \otimes \mathbb{C}^d)$ has $d^2$ eigenvalues, but only $d(d+1)/2$ of them are potentially non-zero: the other $d(d-1)/2$ eigenvalues necessarily equal $0$ since they correspond to eigenvectors in the $(d(d-1)/2)$-dimensional subspace of $\mathbb{C}^d \otimes \mathbb{C}^d$ that is orthogonal to $\mathbb{C}^d \vee \mathbb{C}^d$.

\subsection{Symmetric Separability and Symmetric Entanglement Witnesses}

A symmetric quantum state $\rho \in \mathcal{L}(\mathbb{C}^d \otimes \mathbb{C}^d)$ is separable if and only if there exist vectors $\{\mathbf{v_j}\} \subset \mathbb{C}^d$ such that \cite{Yu16}
\[
	\rho = \sum_{j} \mathbf{v_j}\mathbf{v_j}^* \otimes \mathbf{v_j}\mathbf{v_j}^*.
\]
That is, if $\rho$ is symmetric and separable then the separable decomposition~\eqref{eq:sep_decomp} can be chosen so that $\mathbf{w_j} = \mathbf{v_j}$ for all $j$, so it makes sense to discuss separability of operators in $\mathcal{L}(\mathbb{C}^d \vee \mathbb{C}^d)$.

A \emph{symmetric entanglement witness} is a Hermitian operator $W \in \mathcal{L}(\mathbb{C}^d \vee \mathbb{C}^d)$ with the property that
\[
    (\mathbf{v} \otimes \mathbf{v})^*W(\mathbf{v} \otimes \mathbf{v}) \geq 0 \quad \text{for all} \quad \mathbf{v} \in \mathbb{C}^{d},
\]
or equivalently it is a Hermitian matrix $W = P_{\vee} W P_{\vee} \in \mathcal{L}(\mathbb{C}^d \otimes \mathbb{C}^d)$ with the same property \cite{TG10,Yu16}.

\section{Decomposable Symmetric Entanglement Witnesses}\label{sec:decomp_sew}

We begin by considering symmetric entanglement witnesses of the following special form:

\begin{definition}\label{defn:dec_sew}
    We say that $W \in \mathcal{L}(\mathbb{C}^d \otimes \mathbb{C}^d)$ is a \emph{decomposable symmetric entanglement witness} if there exist positive semidefinite $X,Y \in \mathcal{L}(\mathbb{C}^d \otimes \mathbb{C}^d)$ satisfying $P_{\vee}XP_{\vee} = X$, $P_{\vee}YP_{\vee} = Y$, and
    \[
        W = P_{\vee} X^\Gamma P_{\vee} + Y.
    \]
\end{definition}

It is straightforward to see that if $W$ is a decomposable symmetric entanglement witness then it really is, as its name suggests, a symmetric entanglement witness: for all separable $\mathbf{v} \in \mathbb{C}^d$ we have
\begin{align*}
    (\mathbf{v} \otimes \mathbf{v})^*W(\mathbf{v} \otimes \mathbf{v}) & = (\mathbf{v} \otimes \mathbf{v})^*P_{\vee} X^\Gamma P_{\vee}(\mathbf{v} \otimes \mathbf{v}) + (\mathbf{v} \otimes \mathbf{v})^*Y(\mathbf{v} \otimes \mathbf{v}) \\
    & \geq (\mathbf{v} \otimes \mathbf{v})^* X^\Gamma (\mathbf{v} \otimes \mathbf{v}) \\
    & = \tr\big( X^\Gamma (\mathbf{v}\mathbf{v}^* \otimes \mathbf{v}\mathbf{v}^*) \big) \\
    & = \tr\big( X (\mathbf{v}\mathbf{v}^* \otimes \overline{\mathbf{v}\mathbf{v}^*}) \big) \geq 0.
\end{align*}

However, it is \emph{not} the case that all symmetric entanglement witnesses are decomposable, even in the $d = 2$ case:

\begin{example}\label{exam:wit_not_dec}
    The matrix
    \[
        W = \begin{bmatrix}
            0 & 0 & 0 & -2 \\
            0 & 1 & 1 & 0 \\
            0 & 1 & 1 & 0 \\
            -2 & 0 & 0 & 0
        \end{bmatrix} \in \mathcal{L}(\mathbb{C}^2 \otimes \mathbb{C}^2)
    \]
    is a symmetric entanglement witness since $P_{\vee} W P_{\vee} = W$ and
    \begin{align*}
        (\mathbf{v} \otimes \mathbf{v})^* W (\mathbf{v} \otimes \mathbf{v})
        & = \begin{bmatrix}
            \overline{v_1}^2 & \overline{v_1v_2} & \overline{v_1v_2} & \overline{v_2}^2
        \end{bmatrix} \begin{bmatrix}
            -2v_2^2 \\ 2v_1v_2 \\ 2v_1v_2 \\ -2v_1^2
        \end{bmatrix}\\
        & = -2\overline{v_1}^2v_2^2 + 2(\overline{v_1v_2})(v_1v_2) + 2(\overline{v_1v_2})(v_1v_2) -2v_1^2\overline{v_2}^2 \\
        & = 8\left(\Re{v_1}\Im{v_2} - \Im{v_1}\Re{v_2} \right)^2 \\
        & \geq 0
    \end{align*}
    for all $\mathbf{v} \in \mathbb{C}^2$. However, it is \emph{not} a decomposable symmetric entanglement witness: suppose $X,Y \in \mathcal{L}(\mathbb{C}^d \otimes \mathbb{C}^d)$ are such that $P_{\vee}XP_{\vee} = X$, $P_{\vee}YP_{\vee} = Y$, and $W = P_{\vee}X^\Gamma P_{\vee} + Y$. If $X$ and $Y$ are positive semidefinite, then $0 = w_{1,1} = x_{1,1} + y_{1,1}$ implies $y_{1,1} = 0$, which implies $y_{1,4} = 0$. Then $-2 = w_{1,4} = x_{2,2} + y_{1,4} = x_{2,2}$, which contradicts positive semidefiniteness of $X$.
\end{example}

Example~\ref{exam:wit_not_dec} contrasts with the usual (i.e., non-symmetric) separability problem, where all entanglement witnesses are decomposable in the $d = 2$ case \cite{HHH96}. We could avoid this ``problem'' by removing the requirement that $P_{\vee}XP_{\vee} = X$ from Definition~\ref{defn:dec_sew}. However, we prefer the definition as-is for at least three reasons:

\begin{itemize}
    \item[1)] We want to be able to think of decomposable symmetric entanglement witnesses as being built out of operators in $\mathcal{L}(\Sym{d})$. If $X$ satisfies $P_{\vee}XP_{\vee} = X$ then we can do this: Definition~\ref{defn:dec_sew} could equivalently be rephrased entirely in terms of positive semidefinite operators $X,Y \in \mathcal{L}(\Sym{d})$ that create a symmetric entanglement witness $W \in \mathcal{L}(\Sym{d})$, whereas this is not possible without the $P_{\vee}XP_{\vee} = X$ restriction.
    
    \item[2)] Requiring $P_{\vee}XP_{\vee} = X$ leads to much nicer mathematical structure in general. Essentially none of the results after this point in the paper have ``nice'' analogs if the condition $P_{\vee}XP_{\vee} = X$ is removed. For example, the upcoming Lemma~\ref{lem:real_eigenvalues} describes the eigenvalues of the extreme rays of the convex set of these decomposable symmetric entanglement witnesses, but there is no analogous formula for the extreme rays without the $P_{\vee}XP_{\vee} = X$ condition.
    
    \item[3)] When considering spectral properties (as we do in this paper), it does not seem to matter whether we require $P_{\vee}XP_{\vee} = X$ or not. For example, we will show in the upcoming Theorem~\ref{thm:2x2_eigs} that it has no effect on the possible eigenvalues of the resulting decomposable symmetric entanglement witnesses, at least in the $d = 2$ case.
\end{itemize}

Since our primary interest in these decomposable symmetric entanglement witnesses concerns their spectral properties, we typically just consider the case of Definition~\ref{defn:dec_sew} where $Y = O$. Indeed, adding the positive semidefinite matrix $Y$ to the Hermitian matrix $P_{\vee} X^\Gamma P_{\vee}$ can only increase its eigenvalues \cite[Corollary~4.3.3]{HJ90}, so $Y = O$ is the most extreme and most interesting case.

In order to understand the set of matrices of the form $P_{\vee} X^\Gamma P_{\vee}$ (where $X$ is positive semidefinite) better, we begin by asking what its extreme rays are. A natural first guess is that the extreme rays are the matrices of the form $W = P_{\vee}(\mathbf{v}\mathbf{v}^*)^\Gamma P_{\vee}$, where $\mathbf{v} \in \Sym{d}$. However, it turns out that this is not quite right: the following lemma shows that we can furthermore choose $\mathbf{v}$ to have all of its entries real.

\begin{lemma}\label{lem:real_extreme}
    Suppose $\mathbf{v} \in \Sym{d}$ and $W = P_{\vee}(\mathbf{v}\mathbf{v}^*)^\Gamma P_{\vee}$. Then
    \[
        W = P_{\vee}(\mathbf{x}\mathbf{x}^*)^\Gamma P_{\vee} + P_{\vee}(\mathbf{y}\mathbf{y}^*)^\Gamma P_{\vee},
    \]
    where $\mathbf{x} = \Re{\mathbf{v}}$ and $\mathbf{y} = \Im{\mathbf{v}}$. In particular, the entries of $W$ are all real.
\end{lemma}

\begin{proof}
    If we can prove that the entries of $W$ are all real then we will be done, because that would tell us that the entries of
    \begin{align*}
        W & = P_{\vee}(\mathbf{v}\mathbf{v}^*)^\Gamma P_{\vee} \\
        & = P_{\vee}\big((\mathbf{x} + i\mathbf{y})(\mathbf{x} + i\mathbf{y})^*\big)^\Gamma P_{\vee} \\
        & = P_{\vee}(\mathbf{x}\mathbf{x}^*)^\Gamma P_{\vee} + P_{\vee}(\mathbf{y}\mathbf{y}^*)^\Gamma P_{\vee} + i(P_{\vee}(\mathbf{y}\mathbf{x}^*)^\Gamma P_{\vee} - P_{\vee}(\mathbf{x}\mathbf{y}^*)^\Gamma P_{\vee})
    \end{align*}
    are all real. This would imply that $P_{\vee}(\mathbf{y}\mathbf{x}^*)^\Gamma P_{\vee} = P_{\vee}(\mathbf{x}\mathbf{y}^*)^\Gamma P_{\vee}$ and $W = P_{\vee}(\mathbf{x}\mathbf{x}^*)^\Gamma P_{\vee} + P_{\vee}(\mathbf{y}\mathbf{y}^*)^\Gamma P_{\vee}$, as desired.
    
    To see that the entries of $W$ are indeed all real, let $[A]_{k,\ell;m,n}$ denote the $(m,n)$-entry of the $(k,\ell)$-block of a matrix $A$. Since $\mathbf{v} \in \Sym{d}$, we can write $\mathbf{v} = \sum_j \mathbf{w_j} \otimes \mathbf{w_j}$ for some $\{\mathbf{w_j}\} \subset \mathbb{C}^d$. Then
    \begin{align*}
        (\mathbf{v}\mathbf{v}^*)^\Gamma & = \sum_{i,j=1} (\mathbf{w_i}\mathbf{w_j}^*) \otimes \overline{(\mathbf{w_j}\mathbf{w_i}^*)}, \quad \text{so}\\
        \big[(\mathbf{v}\mathbf{v}^*)^\Gamma\big]_{k,\ell;m,n} & = \sum_{i,j=1} w_{i,k} \overline{w_{j,\ell}w_{j,m}}w_{i,n},
    \end{align*}
    where $w_{i,k}$ represents the $k$-th entry of $\mathbf{w_i}$. Then
    \begin{align*}
        & \big[P_{\vee}(\mathbf{v}\mathbf{v}^*)^\Gamma P_{\vee}\big]_{k,\ell;m,n} \\
        & \qquad = \frac{1}{4} \sum_{i,j=1} \left(
        w_{i,k} \overline{w_{j,\ell}w_{j,m}}w_{i,n} +  
        w_{i,m} \overline{w_{j,\ell}w_{j,k}}w_{i,n} +
         w_{i,m} \overline{w_{j,n}w_{j,k}}w_{i,\ell} +
          w_{i,k} \overline{w_{j,n}w_{j,m}}w_{i,\ell}\right), \\
         & \qquad = \frac{1}{2} \sum_{i,j=1} \left(
        \mathrm{Re}(w_{i,k} \overline{w_{j,\ell}w_{j,m}}w_{i,n}) +  
        \mathrm{Re}(w_{i,m} \overline{w_{j,\ell}w_{j,k}}w_{i,n})\right),
    \end{align*}
    which is clearly real for all $k$, $\ell$, $m$, and $n$.
\end{proof}

Lemma~\ref{lem:real_extreme} is at least a little bit surprising---analogous versions of it are not true if the transpose is replaced by another positive linear map, it is not true without the symmetric projections, and it is not true if $\mathbf{v} \in \Sym{d}$ is replaced by $\mathbf{v} \in \mathbb{C}^d \otimes \mathbb{C}^d$.

In order to begin understanding what the possible eigenvalues of decomposable entanglement witnesses are, we first describe the eigenvalues of the extreme rays described by Lemma~\ref{lem:real_extreme}. To this end, we note that if $\mathbf{v} \in \RSym{d}$ then there exist scalars $\{\alpha_j\} \subset \mathbb{R}$ and an orthonormal basis $\{\mathbf{w_j}\} \subset \mathbb{R}^d$ such that
\begin{align}\label{eq:spectral_decomp_real_sym}
    \mathbf{v} = \sum_{j=1}^d \alpha_j \mathbf{w_j} \otimes \mathbf{w_j}.
\end{align}
Indeed, such a decomposition can be found simply by applying the (real) spectral decomposition to the ``matricization'' of $\mathbf{v}$ (i.e., the matrix that is obtained by placing the entries of $\mathbf{v}$ row-by-row into a $d \times d$ matrix). Indeed, we can decompose this matricization as $\mathrm{mat}(\mathbf{v}) = \sum_{j=1}^d \alpha_j \mathbf{w_j}\mathbf{w_j}^*$, and unravelling it back into a $d^2$-dimensional vector gives us exactly Equation~\eqref{eq:spectral_decomp_real_sym}, which we simply refer to as the spectral decomposition of $\mathbf{v}$ from now on.

With this decomposition in hand, we can now describe explicitly the $\binom{d+1}{2}$ potentially non-zero eigenvalues of the matrix $P_{\vee}(\mathbf{v}\mathbf{v}^*)^\Gamma P_{\vee}$ when $\mathbf{v} \in \RSym{d}$:

\begin{lemma}\label{lem:real_eigenvalues}
    Suppose $\mathbf{v} \in \RSym{d}$ has spectral decomposition as in Equation~\eqref{eq:spectral_decomp_real_sym}. Then the matrix $P_{\vee}(\mathbf{v}\mathbf{v}^*)^\Gamma P_{\vee}$ has eigenvalues
	\begin{align*}
    	\alpha_i\alpha_j & \quad \text{for} \quad 1 \leq i \leq j \leq d,
	\end{align*}
	together with $\binom{d}{2}$ eigenvalues equal to $0$.
\end{lemma}

\begin{proof}
	After using the spectral decomposition~\eqref{eq:spectral_decomp_real_sym}, straightforward computation shows that
	\[
		P_{\vee}(\mathbf{v}\mathbf{v}^*)^\Gamma P_{\vee} = \sum_{i,j=1}^d \alpha_i \alpha_j P_{\vee}(\mathbf{w_i}\mathbf{w_j}^* \otimes \mathbf{w_j}\mathbf{w_i}^*)P_{\vee}.
	\]
	To find the eigenvalues, we first define the following unnnormalized (eigen)vectors:
	\begin{align*}
		\mathbf{x_{i,j}} & := \mathbf{w_i} \otimes \mathbf{w_j} + \mathbf{w_j} \otimes \mathbf{w_i} \quad \text{for} \quad 1 \leq i \leq j \leq d.
	\end{align*}
	Direct computation then shows that
	\begin{align*}
		\big(P_{\vee}(\mathbf{v}\mathbf{v}^*)^\Gamma P_{\vee}\big)\mathbf{x_{i,j}} = \alpha_i\alpha_j \mathbf{x_{i,j}},
	\end{align*}
	so $\alpha_i\alpha_j$ is an eigenvalue of $P_{\vee}(\mathbf{v}\mathbf{v}^*)^\Gamma P_{\vee}$ associated with the eigenvector $\mathbf{x_{i,j}}$. The other $\binom{d}{2}$ eigenvalues equal to $0$ come from the $\binom{d}{2}$-dimensional (antisymmetric) subspace orthogonal to $\Sym{d}$.
\end{proof}

We should emphasize that the eigenvalues $\alpha_i\alpha_j$ described in Lemma~\ref{lem:real_eigenvalues} can be either positive or negative: the $\alpha_j$'s are eigenvalues of the real symmetric (but not necessarily positive semidefinite) matrix $\mathrm{mat}(\mathbf{v})$, so they can have either sign. This is unlike the coefficients in the more common Schmidt decomposition of vectors in $\mathbb{C}^{d_1} \otimes \mathbb{C}^{d_2}$ and Takagi decomposition of vectors in $\Sym{d}$. In the non-symmetric setting, the result analogous to Lemma~\ref{lem:real_eigenvalues} instead uses the vector's Schmidt coefficients \cite[Lemma~1]{JP18}.

Lemma~\ref{lem:real_eigenvalues} can be used to obtain a bound on the number of negative eigenvalues of decomposable symmetric entanglement witnesses of the form $P_{\vee}(\mathbf{v}\mathbf{v}^*)^\Gamma P_{\vee}$ that is tighter than the general bound of $\binom{d}{2} = d(d-1)/2$ that applies to all symmetric entanglement witnesses \cite{LJ20}:

\begin{corollary}\label{cor:num_neg_eigs_from_real_sym_v}
	If $\mathbf{v} \in \RSym{d}$ then the matrix $P_{\vee}(\mathbf{v}\mathbf{v}^*)^\Gamma P_{\vee}$ has at most $\lfloor d^2/4 \rfloor$ negative eigenvalues.
\end{corollary}

\begin{proof}
    If $k$ of the $\{\alpha_j\}$ coefficients from $\mathbf{v}$'s spectral decomposition~\eqref{eq:spectral_decomp_real_sym} are negative, then exactly $k(d-k)$ of the products in the set $\{\alpha_i\alpha_j\}_{i\leq j}$ (i.e., eigenvalues of $P_{\vee}(\mathbf{v}\mathbf{v}^*)^\Gamma P_{\vee}$, by Lemma~\ref{lem:real_eigenvalues}) are negative. Straightforward calculus shows that this quantity is maximized when $k=d/2$ (or as close to it as possible, if $d$ is odd), so $k(d-k) = \lfloor d/2 \rfloor\lceil d/2 \rceil = \lfloor d^2/4 \rfloor$.
\end{proof}

The bound of Corollary~\ref{cor:num_neg_eigs_from_real_sym_v} is tight, but if we relax any of its hypotheses then more negative eigenvalues can be attained. For example, in the $d = 3$ case Corollary~\ref{cor:num_neg_eigs_from_real_sym_v} tells us that a matrix of the form $P_{\vee}(\mathbf{v}\mathbf{v}^*)^\Gamma P_{\vee}$ (where $\mathbf{v} \in \RSym{d}$) has at most $2$ negative eigenvalues. However, it is possible to construct witnesses of this form with $3$ negative eigenvalues (which is necessarily maximal, by the general $\binom{d}{2}$ bound) if we relax any of the hypotheses on $\mathbf{v}$:

\begin{itemize}
    \item If $\mathbf{v} = (0, 1, i, 1, 0, i, i, i, 1) \in \Sym{3}$ then $P_{\vee}(\mathbf{v}\mathbf{v}^*)^\Gamma P_{\vee}$ has $3$ negative eigenvalues equal to $-1$.
    
    \item If $\mathbf{v} = (0, 0, 1, 1, 0, 0, 0, 1, 0) \in \mathbb{R}^3 \otimes \mathbb{R}^3$ then $P_{\vee}(\mathbf{v}\mathbf{v}^*)^\Gamma P_{\vee}$ has $3$ negative eigenvalues equal to $-1/2$.
\end{itemize}

When $d \geq 4$, it becomes less clear how many negative eigenvalues witnesses of the form $P_{\vee}(\mathbf{v}\mathbf{v}^*)^\Gamma P_{\vee}$ can have when we relax the restrictions on $\mathbf{v}$. We have randomly generated a few million vectors of different sizes and thus numerically found vectors $\mathbf{v} \in \mathbb{C}^d \otimes \mathbb{C}^d$ such that, for $d = 2, 3, 4, \ldots$, $P_{\vee}(\mathbf{v}\mathbf{v}^*)^\Gamma P_{\vee}$ has $1$, $3$, $5$, $8$, $12$, $16$, $21$, $26$, $30$, $37$, $\ldots$ negative eigenvalues (in fact, we have found these same quantities of negative eigenvalues even just with $\mathbf{v} \in \mathbb{R}^d \otimes \mathbb{R}^d$ or $\mathbf{v} \in \Sym{d}$). However, these numbers just give lower bounds on how many negative eigenvalues $P_{\vee}(\mathbf{v}\mathbf{v}^*)^\Gamma P_{\vee}$ can have. The general $\binom{d}{2} = d(d-1)/2$ upper bound of \cite{LJ20} applies in this setting, and Figure~\ref{fig:dimensional_eigenvalue_comparisons} compares these various lower and upper bounds.

\begin{figure}[!htb]
    \centering
    \begin{tikzpicture}[scale=0.08]
        \foreach \x in {1,2,...,11} {
            \draw[ultra thin,gray] (10*\x, -2) -- (10*\x, 57);
            \node[anchor=north] at (10*\x, -1) {$\x$};
            \draw[thick] (10*\x, -0.5) -- (10*\x, 0.5);
        }
        \foreach \y in {5,10,...,55} {
            \draw[ultra thin,gray] (5, \y) -- (115, \y);
        }
        \foreach \y in {10,20,...,50} {
            \node[anchor=east] at (5, \y) {$\y$};
        }
        \draw[thick] (115, 0) -- coordinate (x axis mid) (5,0);
        \node[anchor=east] at (5,0) {$0$};
        \node[below=0.8cm] at (x axis mid) {Dimension ($d$)};
        \node[rotate=90, above=0.8cm] at (5, 27) {Number of Negative Eigenvalues};

        \node[anchor=north west, shift={(-0.4,0.3)},text opacity=1,fill opacity=0.75,fill=white,inner sep=0.2] at (10*9, 20) {\begin{tabular}{c} $g(d) = \lfloor d^2/4 \rfloor$ \\ Maximum when $X = \mathbf{v}\mathbf{v}^*$ \\ with $\mathbf{v} \in \RSym{d}$ (Corollary~\ref{cor:num_neg_eigs_from_real_sym_v}) \end{tabular}};
        
        \coordinate (cur) at (10*1,0);
        \foreach \x/\y in {1/0, 2/1, 3/2, 4/4, 5/6, 6/9, 7/12, 8/16, 9/20, 10/25, 11/30} {
            \draw (cur) -- (10*\x, \y);
            \coordinate (cur) at (10*\x, \y);
            \node[diamond, fill, inner sep=1pt] (cur) at (10*\x, \y) {};
        }
        \coordinate (cur) at (10*1,0);
        \foreach \x/\y in {1/0, 2/1, 3/3, 4/6, 5/10, 6/15, 7/21, 8/28, 9/36, 10/45, 11/55} {
            \draw (cur) -- (10*\x, \y);
            \coordinate (cur) at (10*\x, \y);
            \node[circle, fill, inner sep=1pt] (cur) at (10*\x, \y) {};
        }
        \node[anchor=south east, shift={(1.2,0.95)},text opacity=1,fill opacity=0.75,fill=white,inner sep=0.2] at (10*7, 21) {\begin{tabular}{c} Maximum over all \\ PSD $X$ (Theorem~\ref{thm:dsew_max_neg_evals}) \\ $f(d) = d(d-1)/2$ \end{tabular}};
        \coordinate (cur) at (10*1,0);
        \foreach \x/\y in {1/0, 2/1, 3/3, 4/5, 5/8, 6/12, 7/16, 8/21, 9/26, 10/30, 11/37} {
            \draw[dashed] (cur) -- (10*\x, \y);
            \coordinate (cur) at (10*\x, \y);
            \node[circle, draw, inner sep=1pt] (cur) at (10*\x, \y) {};
        }
        \node[anchor=west,shift={(0,0.3)}] at (10*11, 37) {\begin{tabular}{c} Numerical maximum when \\ $X = \mathbf{v}\mathbf{v}^*$ with $\mathbf{v} \in \mathbb{C}^d \otimes \mathbb{C}^d$ \end{tabular}};
    \end{tikzpicture}

    \caption{The maximum number of negative eigenvalues of decomposable symmetric entanglement witnesses (i.e., operators of the form $P_{\vee} X^\Gamma P_{\vee} + Y$) of various forms, by dimension.}\label{fig:dimensional_eigenvalue_comparisons}
\end{figure}
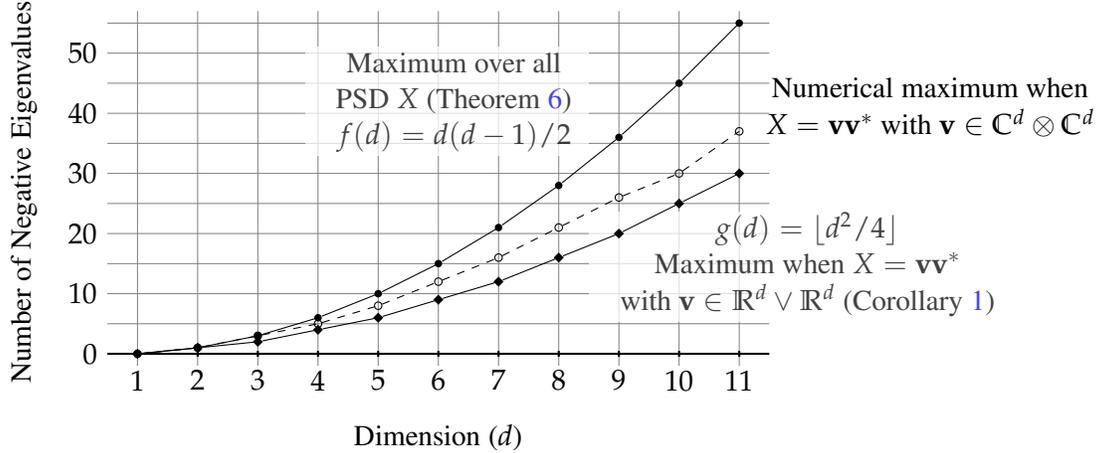

We will show in Theorem~\ref{thm:dsew_max_neg_evals} that matrices of the form $P_{\vee} X^\Gamma P_{\vee}$, where $X \in \mathcal{L}(\mathbb{C}^d \otimes \mathbb{C}^d)$ is positive semidefinite but not necessarily of the form $X = \mathbf{v}\mathbf{v}^*$, saturate the upper bound---they can have up to $d(d-1)/2$ negative eigenvalues.

%

\section{Absolute Symmetric Separability}\label{sec:abs_sep}

We now consider the problem of characterizing the following set of symmetric quantum states:

\begin{definition}\label{defn:sym_abs_sep}
    We say that $\rho \in \mathcal{L}(\Sym{d})$ is \emph{absolutely symmetric separable} if $U\rho U^*$ is separable for all unitaries $U \in \mathcal{L}(\Sym{d})$.
\end{definition}

By the spectral decomposition, the above definition really is a spectral one: $\rho \in \mathcal{L}(\Sym{d})$ is absolutely symmetric separable if and only if every symmetric quantum state $\sigma \in \mathcal{L}(\Sym{d})$ with the same spectrum as $\rho$ is separable. The non-symmetric version of this problem has been studied extensively \cite{KZ00,Joh13}, but essentially nothing is known about it beyond what the partial transpose tells us \cite{Hil07,AJR15}.

With that in mind, we consider the following (presumably) easier-to-characterize set of symmetric states:

\begin{definition}\label{defn:sym_abs_ppt}
    We say that $\rho \in \mathcal{L}(\Sym{d})$ is \emph{absolutely symmetric PPT} if $P_{\vee}(U\rho U^*)^\Gamma P_{\vee}$ is positive semidefinite for all unitaries $U \in \mathcal{L}(\Sym{d})$.
\end{definition}

Since every separable state has positive partial transpose, every absolutely symmetric separable state is absolutely symmetric PPT. In fact, every absolutely symmetric separable state satisfies the seemingly stronger requirement that $(U\rho U^*)^\Gamma$ is positive semidefinite for all unitaries $U \in \mathcal{L}(\Sym{d})$. Much like we could have altered Definition~\ref{defn:dec_sew} to not require $P_{\vee}XP_{\vee} = X$, we could alter Definition~\ref{defn:sym_abs_ppt} to require that $(U\rho U^*)^\Gamma$ (not just $P_{\vee}(U\rho U^*)^\Gamma P_{\vee}$) is positive semidefinite for all unitaries $U \in \mathcal{L}(\Sym{d})$. However, we again prefer the current definition for several reasons:

\begin{itemize}
    \item[1)] The version of Definition~\ref{defn:sym_abs_ppt} that we used agrees with the philosophy of considering operators only as they act on $\mathcal{L}(\Sym{d})$: whether or not $P_{\vee}(U\rho U^*)^\Gamma P_{\vee}$ is positive semidefinite depends only on how it acts on symmetric vectors, whereas positive semidefiniteness of $(U\rho U^*)^\Gamma$ depends on its action on non-symmetric vectors as well.
    
    \item[2)] Positive semidefiniteness of $P_{\vee}(U\rho U^*)^\Gamma P_{\vee}$ is implied by positive semidefiniteness of $(U\rho U^*)^\Gamma$, so a characterization of the quantum states of Definition~\ref{defn:sym_abs_ppt} automatically implies bounds on states with the property that $(U\rho U^*)^\Gamma$ is positive semidefinite for all unitaries $U \in \mathcal{L}(\Sym{d})$. We establish such a characterization in the upcoming Theorem~\ref{thm:asymppt_from_eigenvalues}.
    
    \item[3)] The distinction between these two sets (at least in the ``absolute'' setting where we use \emph{all} unitaries $U \in \mathcal{L}(\Sym{d})$) does not seem to matter, at least in the $d = 2$ case (see the upcoming Theorem~\ref{thm:asymppt_from_eigenvalues_2dim}).
\end{itemize}

The set of (not-necessarily-symmetric) absolutely PPT states were characterized completely in \cite{Hil07}, and the main result of this section provides an analogous characterization in the symmetric case. In order to state this characterization succinctly, we first define a \emph{upper-triangular matricization} of a vector $\mathbf{x} \in \mathbb{R}^{d(d+1)/2}$ to be any upper-triangular matrix whose upper-triangular entries are the entries of $\mathbf{x}$, in some order. For example, if $d = 2$ and $\mathbf{x} = (1, 3, 5)$ then
\begin{align}\label{eq:triangular_matricization}
    \begin{bmatrix}
        1 & 3 \\ 0 & 5
    \end{bmatrix}, \quad \begin{bmatrix}
        1 & 5 \\ 0 & 3
    \end{bmatrix}, \quad \text{and} \quad \begin{bmatrix}
        3 & 1 \\ 0 & 5
    \end{bmatrix}
\end{align}
are all upper-triangular matricizations of $\mathbf{x}$. We similarly define a \emph{symmetric matricization} of $\mathbf{x} \in \mathbb{R}^{d(d+1)/2}$ to be any matrix of the form $X + X^T$, where $X$ is an upper-triangular matricization of $\mathbf{x}$. For example, if $d = 2$ and $\mathbf{x} = (1, 3, 5)$ then
\[
    \begin{bmatrix}
        2 & 3 \\ 3 & 10
    \end{bmatrix}, \quad \begin{bmatrix}
        2 & 5 \\ 5 & 6
    \end{bmatrix}, \quad \text{and} \quad \begin{bmatrix}
        6 & 1 \\ 1 & 10
    \end{bmatrix}
\]
are the symmetric matricizations of $\mathbf{x}$ corresponding to the upper-triangular matricizations~\eqref{eq:triangular_matricization}.

\begin{theorem}\label{thm:asymppt_from_eigenvalues}
    Suppose $\rho \in \mathcal{L}(\Sym{d})$ is a symmetric quantum state, and let $\bm{\lambda} = (\lambda_1,\lambda_2,\ldots,\lambda_{d(d+1)/2})$ be a vector containing its $\binom{d+1}{2} = d(d+1)/2$ eigenvalues. Then $\rho$ is absolutely symmetric PPT if and only if every symmetric matricization of $\bm{\lambda}$ is positive semidefinite.
\end{theorem}

\begin{proof}
    The state $\rho$ being absolutely symmetric PPT means that $P_{\vee}(U\rho U^*)^\Gamma P_{\vee}$ is positive semidefinite for all unitaries $U \in \mathcal{L}(\Sym{d})$, which is equivalent to
    \[
        \mathbf{v}^* (U \rho U^*)^\Gamma \mathbf{v} \geq 0
    \]
    for all unitaries $U \in \mathcal{L}(\Sym{d})$ and all $\mathbf{v} \in \Sym{d}$. Rearranging a bit shows that
    \begin{align*}
        \mathbf{v}^* (U \rho U^*)^\Gamma \mathbf{v} & = \tr\big( (U \rho U^*) (\mathbf{v}\mathbf{v}^*)^\Gamma \big) = \tr\big( (U \rho U^*) (P_{\vee}(\mathbf{v}\mathbf{v}^*)^\Gamma P_{\vee}) \big),
    \end{align*}
    where the last equality follows from the fact that $U \rho U^*$ is supported on the symmetric subspace. It follows that $\rho$ is absolutely symmetric PPT if and only if
    \begin{align}\label{eq:asymppt_trace}
        \tr\big( (U \rho U^*) (P_{\vee}(\mathbf{v}\mathbf{v}^*)^\Gamma P_{\vee}) \big) \geq 0
    \end{align}
    for all unitaries $U \in \mathcal{L}(\Sym{d})$ and all $\mathbf{v} \in \Sym{d}$. By Lemma~\ref{lem:real_extreme}, it furthermore suffices to require that $\mathbf{v} \in \RSym{d}$.
    
    If we denote the eigenvalues of $P_{\vee}(\mathbf{v}\mathbf{v}^*)^\Gamma P_{\vee}$ by $\mu_1 \geq \mu_2 \geq \cdots \geq \mu_{d(d+1)/2}$, then \cite[Problem~III.6.14]{Bha97} tells us that minimizing the value on the left-hand-side of Inequality~\eqref{eq:asymppt_trace} over all unitaries $U \in \mathcal{L}(\Sym{d})$ is equivalent to minimizing
    \begin{align}\label{eq:eig_sort_minimize}
        \sum_{j=1}^{d(d+1)/2} \lambda_{\pi(j)} \mu_{j}
    \end{align}
    over all permutations $\pi : \{1,2,\ldots, d(d+1)/2\} \rightarrow \{1,2,\ldots, d(d+1)/2\}$. Well, Lemma~\ref{lem:real_eigenvalues} tells us that the eigenvalues of $P_{\vee}(\mathbf{v}\mathbf{v}^*)^\Gamma P_{\vee}$ are all of the form $\alpha_i \alpha_j$ for some family of real numbers $\{\alpha_j\}_{j=1}^d \subset \mathbb{R}$. It follows that $\rho$ is absolutely symmetric PPT if and only if
    \begin{align}\label{eq:abs_sym_ppt_ineq}
        \sum_{i=1}^{d}\sum_{j=i}^d \alpha_i \alpha_j \lambda_{f(i,j)} \geq 0
    \end{align}
    for all bijections $f : \{(i,j) : 1 \leq i \leq j \leq d\} \rightarrow \{1, 2, \ldots, d(d+1)/2\}$ and all $\{\alpha_j\}_{j=1}^d \subset \R$. By grouping these numbers together into $d \times d$ matrices, we see that Inequality~\eqref{eq:abs_sym_ppt_ineq} is equivalent to
    \begin{align}\label{eq:abs_sym_ppt_ineq2}
        \bm{\alpha}^* \Lambda \bm{\alpha} \geq 0
    \end{align}
    for all symmetric matricizations $\Lambda$ of $\bm{\lambda}$ (the order in which the entries of $\bm{\lambda}$ are put into the matrix $\Lambda$ is determined by the bijection $f$), where $\bm{\alpha} = (\alpha_1,\alpha_2,\ldots,\alpha_d) \in \mathbb{R}^d$. Since Inequality~\eqref{eq:abs_sym_ppt_ineq2} holds for all vectors $\bm{\alpha} \in \mathbb{R}^d$, it is equivalent to positive semidefiniteness of each $\Lambda$. That is, $\rho$ is absolutely symmetric PPT if and only if all symmetric matricizations of $\bm{\lambda}$ are positive semidefinite, which is what we wanted to show.
\end{proof}

Since absolutely symmetric separable states are absolutely symmetric PPT, the characterization of Theorem~\ref{thm:asymppt_from_eigenvalues} provides (seemingly quite strong) necessary conditions for absolute symmetric separability. It also leads to some simple corollaries that would not otherwise be obvious. For example, every absolutely symmetric separable quantum state has full rank:

\begin{corollary}\label{cor:no_zero_eigenvalues}
    If a quantum state $\rho \in \mathcal{L}(\Sym{d})$ is absolutely symmetric separable (or even just absolutely symmetric PPT) then all $d(d+1)/2$ of its eigenvalues are strictly positive.
\end{corollary}

\begin{proof}
    One of the symmetric matricizations of $\bm{\lambda}$ (the vector containing the eigenvalues of $\rho$) has top-left $2 \times 2$ block
    \[
        \begin{bmatrix}
            2\lambda_{d(d+1)/2} & \lambda_1 \\
            \lambda_1 & 2\lambda_2
        \end{bmatrix}.
    \]
    If any of the eigenvalues of $\rho$ equalled $0$ then we would have $\lambda_{d(d+1)/2} = 0$, so positive semidefiniteness of this symmetric matricization would imply $\lambda_1 = 0$, so $\bm{\lambda} = 0$, which contradicts the fact that $\tr(\rho) = \sum_j \lambda_j = 1$.
\end{proof}

For the sake of comparison, a similar result holds for non-symmetric absolute separability, but in that case it is possible to have a single eigenvalue that equals $0$ \cite[Proposition~1]{AJR15}.

\subsection{Low-Dimensional Simplifications}\label{sec:abs_sep_small_dim}

While using Theorem~\ref{thm:asymppt_from_eigenvalues} na\"{i}vely to check whether or not a state is absolutely symmetric PPT relies on checking positive semidefiniteness of all $(d(d+1)/2)!$ symmetric matricizations of its eigenvalues, that set of matrices can be reduced considerably by exploiting symmetries and ignoring cases that are ``obviously'' positive semidefinite.

For example, if $d = 2$ then there are $3! = 6$ symmetric matricizations, but we do not need to check positive semidefiniteness of them all: if we denote the eigenvalues by $\lambda_1 \geq \lambda_2 \geq \lambda_3 \geq 0$ then the $3! = 6$ symmetric matricizations are
\begin{align}\label{eq:lambda_2dim_sym}
    \Lambda_1 & := \begin{bmatrix}
        2\lambda_1 & \lambda_2 \\ \lambda_2 & 2\lambda_3
    \end{bmatrix}, & \Lambda_2 & := \begin{bmatrix}
        2\lambda_1 & \lambda_3 \\ \lambda_3 & 2\lambda_2
    \end{bmatrix}, & \Lambda_3 & := \begin{bmatrix}
        2\lambda_2 & \lambda_1 \\ \lambda_1 & 2\lambda_3
    \end{bmatrix},\nonumber\\
    \Lambda_4 & := \begin{bmatrix}
        2\lambda_2 & \lambda_3 \\ \lambda_3 & 2\lambda_1
    \end{bmatrix}, & \Lambda_5 & := \begin{bmatrix}
        2\lambda_3 & \lambda_1 \\ \lambda_1 & 2\lambda_2
    \end{bmatrix}, & \Lambda_6 & := \begin{bmatrix}
        2\lambda_3 & \lambda_2 \\ \lambda_2 & 2\lambda_1
    \end{bmatrix}.
\end{align}
Of these, $\Lambda_2$ and $\Lambda_4$ are necessarily positive semidefinite (since they are diagonally dominant), so we do not need to check them. Furthermore, permuting rows and columns of these matrices show that $\Lambda_1$ is positive semidefinite if and only if $\Lambda_6$ is, and $\Lambda_3$ is positive semidefinite if and only if $\Lambda_5$ is. Finally, positive semidefiniteness of $\Lambda_5$ is equivalent to $\lambda_1 \leq 2\sqrt{\lambda_2\lambda_3}$, which implies $\lambda_2 \leq 2\sqrt{\lambda_1\lambda_3}$, which is equivalent to positive semidefiniteness of $\Lambda_1$. Altogether, this shows that if $\Lambda_5$ is positive semidefinite then so is $\Lambda_j$ for each $1 \leq j \leq 6$.

The observation of the previous paragraph is stated more succinctly in the following theorem, which also establishes the fact that this same characterization works for absolutely symmetric separability (not just absolutely symmetric positive partial transposition).

\begin{theorem}\label{thm:asymppt_from_eigenvalues_2dim}
    Suppose $\rho \in \mathcal{L}(\Sym{2})$ has eigenvalues $\lambda_1 \geq \lambda_2 \geq \lambda_3 \geq 0$. The following are equivalent:
    
    \begin{enumerate}[(a)]
        \item $\rho$ is absolutely symmetric separable;
        
        \item $\rho$ is absolutely symmetric PPT; and
        
        \item $\lambda_1 \leq 2\sqrt{\lambda_2\lambda_3}$.
    \end{enumerate}
\end{theorem}

\begin{proof}
    We showed the equivalence of conditions~(b) and~(c) just before the statement of this theorem, and we discussed why (a) implies~(b) right after Definition~\ref{defn:sym_abs_ppt}.
    
    In order to prove the only remaining implication (that~(c) implies~(a)), we note that it suffices to show that $C(\rho) \leq \lambda_1 - 2\sqrt{\lambda_2\lambda_3}$ for all symmetric states $\rho \in \mathcal{L}(\Sym{2})$, where $C(\rho)$ is the (bosonic) concurrence of $\rho$ \cite{ESBL02}. Indeed, the concurrence of $\rho$ equals $0$ if and only if $\rho$ is separable, so if (c) holds then we would have $C(\rho) \leq \lambda_1 - 2\sqrt{\lambda_2\lambda_3} \leq 0$, so $\rho$ is separable (and in fact, absolutely separable, since this argument applied to every symmetric quantum state with the same eigenvalues).
    
    To prove this inequality, we start with the formula $C(\rho) = \max\{0, \nu_1 - \nu_2 - \nu_3\}$, where $\nu_1 \geq \nu_2 \geq \nu_3 \geq 0$ are the non-negative square roots of the (necessarily real and non-negative) eigenvalues of $\rho S \overline{\rho} S$ \cite{ESBL02} and $S \in \mathcal{L}(\Sym{2})$ is the linear unitary operator that, in the orthonormal basis $\{\mathbf{e_1} \otimes \mathbf{e_1}, (\mathbf{e_1} \otimes \mathbf{e_2} + \mathbf{e_2} \otimes \mathbf{e_1})/\sqrt{2}, \mathbf{e_2} \otimes \mathbf{e_2}\}$ has matrix representation
    \[
        S = \begin{bmatrix}
            0 & 0 & 1 \\
            0 & -1 & 0 \\
            1 & 0 & 0
        \end{bmatrix}.
    \]
    Our goal then becomes to show that $\nu_1 - \nu_2 - \nu_3 \leq \lambda_1 - 2\sqrt{\lambda_2\lambda_3}$.
    
    To this end, we proceed in much the same way as was done in \cite{VAD01}. We first recall that if $A$ and $B$ are square matrices then $AB$ has the same eigenvalues as $BA$. If $\rho$ has spectral decomposition $\rho = UDU^*$ then this tells us that $\rho S \overline{\rho} S = UDU^*S\overline{U}DU^TS$ has the same eigenvalues as $\sqrt{D}U^*S\overline{U}DU^TSU\sqrt{D} = (\sqrt{D}U^TSU\sqrt{D})^*(\sqrt{D}U^TSU\sqrt{D})$. It follows that $\nu_1$, $\nu_2$, and $\nu_3$ are exactly the singular values of $\sqrt{D}U^TSU\sqrt{D}$.
    
    We now make use of the fact that if $A$ and $B$ are square matrices then the singular values of $A$, $B$, and $AB$ (ordered so that $\sigma_1 \geq \sigma_2 \geq \cdots$) satisfy the following inequalities (see \cite[Theorem~2.3.5]{JohALA} for the first, which is well-known, and \cite{WX97} for the second):
    \begin{align*}
        \sigma_1(AB) & \leq \sigma_1(A)\sigma_1(B) \quad \text{and} \\
        \sigma_2(AB) + \sigma_3(AB) & \geq \sigma_2(A)\sigma_3(B) + \sigma_3(A)\sigma_2(B).
    \end{align*}
    If we set $A = \sqrt{D}$ and $B = U^TSU\sqrt{D}$ then unitarity of $U^TSU$ tells us that $A$ and $B$ each have singular values $\sqrt{\lambda_1} \geq \sqrt{\lambda_2} \geq \sqrt{\lambda_3} \geq 0$. By subtracting the inequality $\nu_2 + \nu_3 = \sigma_2(AB) + \sigma_3(AB) \geq \sigma_2(A)\sigma_3(B) + \sigma_3(A)\sigma_2(B) = 2\sqrt{\lambda_2\lambda_3}$ from the inequality $\nu_1 = \sigma_1(AB) \leq \sigma_1(A)\sigma_1(B) = \lambda_1$, we get exactly
    \[
        \nu_1 - (\nu_2 + \nu_3) \leq \lambda_1 - 2\sqrt{\lambda_2\lambda_3},
    \]
    as desired.
\end{proof}

Theorem~\ref{thm:asymppt_from_eigenvalues_2dim} is analogous to the non-symmetric result that a quantum state $\rho \in \mathcal{L}(\mathbb{C}^2 \otimes \mathbb{C}^2)$ is absolutely separable if and only if its eigenvalues $\lambda_1 \geq \lambda_2 \geq \lambda_3 \geq \lambda_4$ satisfy $\lambda_1 \leq \lambda_3 + 2\sqrt{\lambda_2\lambda_4}$ \cite{VAD01}. However, while it is obvious that the sets of (non-symmetric) absolutely separable states and absolutely PPT states are the same in the $d = 2$ case (since separability and PPT are the same in this case), the same equivalence in this symmetric setting, as established by Theorem~\ref{thm:asymppt_from_eigenvalues_2dim}, is perhaps surprising. After all, there are simple examples that show that even though \emph{absolute} symmetric separability is the same as absolute symmetric PPT when $d = 2$, their non-absolute counterparts are not equivalent:

\begin{example}\label{exam:sym_sep_ppt}
    Suppose $\mathbf{v} = (1,0,0,1)/\sqrt{2} \in \Sym{2}$ and $\rho = \mathbf{v}\mathbf{v}^*$. Then $\rho^\Gamma$ has $-1/2$ as an eigenvalue, so $\rho$ is not separable. However, $P_{\vee}\rho^\Gamma P_{\vee} = P_{\vee}/2$ is positive semidefinite.
\end{example}

We now begin discussing how to reduce the number of symmetric matricizations that need to be considered in Theorem~\ref{thm:asymppt_from_eigenvalues} when $d \geq 3$. We note that the quantity~\eqref{eq:eig_sort_minimize} is minimized when the permutation $\pi$ is chosen so that the vectors $(\lambda_{\pi(1)},\lambda_{\pi(2)},\ldots,\lambda_{\pi(d(d+1)/2)})$ and $(\mu_1,\mu_2,\ldots,\mu_{d(d+1)/2})$ are sorted in the opposite order. It thus suffices to only consider permutations $\pi$ (and thus bijections $f$ from the same proof, and thus symmetric matricizations) corresponding to the different possible orderings of the numbers $\{\alpha_i\alpha_j\}$, subject to the constraint $\alpha_1 \geq \alpha_2 \geq \cdots \geq \alpha_d$ (this same argument was used in \cite{Hil07} to reduce the number of linear matrix inequalities that needed to be considered).

We can furthermore multiply $\bm{\alpha} = (\alpha_1,\alpha_2,\ldots,\alpha_d)$ by $-1$, since that does not change the quantity $\bm{\alpha}^* \Lambda \bm{\alpha}$ used in the proof of Theorem~\ref{thm:asymppt_from_eigenvalues}, allowing us to assume without loss of generality that $\alpha_1 \geq |\alpha_d|$ (and thus in particular that $\alpha_1 \geq 0$), reducing the number of orderings that we need to consider even further. For example, when $d = 2$, if $\alpha_1 \geq \alpha_2$ are real numbers with $\alpha_1 \geq 0$ then there are only two possible orderings of the numbers $\alpha_1^2$, $\alpha_1\alpha_2$, and $\alpha_2^2$: we could have either
\[
    \alpha_1^2 \geq \alpha_1\alpha_2 \geq \alpha_2^2 \quad \text{or} \quad \alpha_1^2 \geq \alpha_2^2 \geq \alpha_1\alpha_2,
\]
depending on whether $\alpha_2 \geq 0$ or $\alpha_2 < 0$. These orderings correspond to the symmetric matricizations $\Lambda_6$ and $\Lambda_5$ from Equation~\eqref{eq:lambda_2dim_sym}, respectively. However, we also do not need to consider the case when each $\alpha_j$ is non-negative, since in that case $\bm{\alpha}^* \Lambda \bm{\alpha} \geq 0$ comes for free. In the $d = 2$ case, we thus just need the single ordering $\alpha_1^2 \geq \alpha_2^2 \geq \alpha_1\alpha_2$, corresponding to the single symmetric matricization $\Lambda_5$, as we already noted in the discussion right before Theorem~\ref{thm:asymppt_from_eigenvalues_2dim}.

These same ideas show that in the $d = 3$ case, we do not have to consider all $(d(d+1)/2)! = 6! = 720$ symmetric matricizations of Theorem~\ref{thm:asymppt_from_eigenvalues} to determine whether or not a quantum state $\rho \in \mathcal{L}(\Sym{3})$ is absolutely symmetric PPT. Rather, it suffices to consider only the symmetric matricizations corresponding to the possible orderings of the $6$ real numbers $\{\alpha_i\alpha_j\}$ subject to the constrains that $\alpha_1 \geq \alpha_2 \geq \alpha_3$, $\alpha_1 \geq |\alpha_3|$, and $\alpha_3 < 0$. There are four such orderings:
\begin{align*}
    \alpha_1^2 & \geq \alpha_3^2 \geq \alpha_2\alpha_3 \geq \alpha_2^2 \geq \alpha_1\alpha_2 \geq \alpha_1\alpha_3, \\
    \alpha_1^2 & \geq \alpha_3^2 \geq \alpha_1\alpha_2 \geq \alpha_2^2 \geq \alpha_2\alpha_3 \geq \alpha_1\alpha_3, \\
    \alpha_1^2 & \geq \alpha_1\alpha_2 \geq \alpha_3^2 \geq \alpha_2^2 \geq \alpha_2\alpha_3 \geq \alpha_1\alpha_3, \\
    \alpha_1^2 & \geq \alpha_1\alpha_2 \geq \alpha_2^2 \geq \alpha_3^2 \geq \alpha_2\alpha_3 \geq \alpha_1\alpha_3,
\end{align*}
corresponding to the four symmetric matricizations
\begin{align}\label{eq:3x3_sym_matrx}
    \Lambda_1 & := \begin{bmatrix}
        2\lambda_6 & \lambda_2 & \lambda_1 \\
        \lambda_2 & 2\lambda_3 & \lambda_4 \\
        \lambda_1 & \lambda_4 & 2\lambda_5
    \end{bmatrix}, & \Lambda_2 & := \begin{bmatrix}
        2\lambda_6 & \lambda_4 & \lambda_1 \\
        \lambda_4 & 2\lambda_3 & \lambda_2 \\
        \lambda_1 & \lambda_2 & 2\lambda_5
    \end{bmatrix},\nonumber\\
    \Lambda_3 & := \begin{bmatrix}
        2\lambda_6 & \lambda_5 & \lambda_1 \\
        \lambda_5 & 2\lambda_3 & \lambda_2 \\
        \lambda_1 & \lambda_2 & 2\lambda_4
    \end{bmatrix}, & \Lambda_4 & := \begin{bmatrix}
        2\lambda_6 & \lambda_5 & \lambda_1 \\
        \lambda_5 & 2\lambda_4 & \lambda_2 \\
        \lambda_1 & \lambda_2 & 2\lambda_3
    \end{bmatrix},
\end{align}
respectively. However, we now show that we do not even need to check positive semidefinite of all four of these matrices---it is enough to only check positive semidefiniteness of $\Lambda_1$.

\begin{corollary}\label{cor:asymppt_from_eigenvalues_3dim}
    Suppose $\rho \in \mathcal{L}(\Sym{3})$ has eigenvalues $\lambda_1 \geq \lambda_2 \geq \cdots \geq \lambda_6 \geq 0$. Then $\rho$ is absolutely symmetric PPT if and only if the matrix
    \[
        \begin{bmatrix}
            2\lambda_6 & \lambda_2 & \lambda_1 \\
            \lambda_2 & 2\lambda_3 & \lambda_4 \\
            \lambda_1 & \lambda_4 & 2\lambda_5
        \end{bmatrix}
    \]
    is positive semidefinite.
\end{corollary}

\begin{proof}
    We just need to show that positive semidefiniteness of the matrix from the statement of the theorem, which is $\Lambda_1$ from Equation~\eqref{eq:3x3_sym_matrx}, implies positive semidefiniteness of $\Lambda_2$, $\Lambda_3$, and $\Lambda_4$ from Equation~\eqref{eq:3x3_sym_matrx}.
    
    Positive semidefiniteness of $\Lambda_1$ implies (by looking at its outermost $2 \times 2$ minor) that $4\lambda_6\lambda_5 \geq \lambda_1^2$, which in turn implies positive semidefiniteness of all $2 \times 2$ minors of $\Lambda_2$, $\Lambda_3$, and $\Lambda_4$. By Sylvester's criterion, it thus suffices to show that positive semidefiniteness of $\Lambda_1$ implies $\det(\Lambda_2), \det(\Lambda_3), \det(\Lambda_4) \geq 0$, which we now demonstrate.
    
    \begin{itemize}
        \item[i)] $\det(\Lambda_4)-\det(\Lambda_3) = 2\lambda_4\lambda_5^2+2\lambda_1^2\lambda_3-2\lambda_3\lambda_5^2-2\lambda_1^2\lambda_4=2(\lambda_1^2-\lambda_5^2)(\lambda_3-\lambda_4) \geq 0$ (with non-negativity following from the fact that $\lambda_1 \geq \lambda_5$ and $\lambda_3 \geq \lambda_4$), so $\det(\Lambda_4) \geq \det(\Lambda_3)$.
        
        \item[ii)] $\det(\Lambda_3)-\det(\Lambda_2)=8\lambda_3\lambda_4\lambda_6-8\lambda_3\lambda_5\lambda_6+2\lambda_1\lambda_2\lambda_5-2\lambda_1\lambda_2\lambda_4+\lambda_4^2\lambda_5-\lambda_4\lambda_5^2=(4\lambda_3\lambda_6+\lambda_4\lambda_5-\lambda_1\lambda_2)(\lambda_4-\lambda_5)$. Since we are assuming that $\Lambda_1$ is positive semidefinite, we know that $4\lambda_5\lambda_6 \geq \lambda_1^2$, so $4\lambda_3\lambda_6 \geq \lambda_1\lambda_2$, so $4\lambda_3\lambda_6 + \lambda_4\lambda_5 - \lambda_1\lambda_2 \geq 0$, so $\det(\Lambda_3) \geq \det(\Lambda_2)$.
        
        \item[iii)] $\det(\Lambda_2)-\det(\Lambda_1)=2\lambda_4^2\lambda_6+2\lambda_2^2\lambda_5-2\lambda_2^2\lambda_6-2\lambda_4^2\lambda_5=2(\lambda_5-\lambda_6)(\lambda_2^2-\lambda_4^2) \geq 0$, so $\det(\Lambda_2) \geq \det(\Lambda_1)$.
    \end{itemize}
    It follows that $\det(\Lambda_4) \geq \det(\Lambda_3) \geq \det(\Lambda_2) \geq \det(\Lambda_1) \geq 0$ whenever $\Lambda_1$ is positive semidefinite, which completes the proof.
\end{proof}

It is worth noting that Corollary~\ref{cor:asymppt_from_eigenvalues_3dim} contrasts with the non-symmetric absolute PPT problem where, in the $d = 3$ case, positive semidefiniteness of two different $3 \times 3$ matrices needs to be checked \cite{Hil07}.

When $d \geq 4$, it becomes less clear what a minimal set of positive semidefinite matrices for characterizing the set of absolutely symmetric PPT states is. If we make the same restrictions to the $\bm{\alpha}$ vectors described just before Corollary~\ref{cor:asymppt_from_eigenvalues_3dim} (i.e., $\alpha_1 \geq \alpha_2 \geq \cdots \geq \alpha_d$, $\alpha_1 \geq |\alpha_d|$, and $\alpha_d < 0$), then we find that for $d = 2, 3, 4, 5, \ldots$ there are $1, 4, 26, 330, \ldots$ possible orderings of the $d(d+1)/2$ products $\{\alpha_i\alpha_j\}$, and we thus only need to check positive semidefiniteness of this many $d \times d$ matrices.


However, we saw in Corollary~\ref{cor:asymppt_from_eigenvalues_3dim} that we do not need all $4$ of these matrices in the $d = 3$ case: it suffices to check just positive semidefiniteness of just one matrix. Similarly, in the $d = 4$ case we expect that we do not need all $26$ of these matrices, but rather it suffices to check positive semidefiniteness of just $4$:

\begin{conjecture}\label{conj:d4_appt}
    Suppose $\rho \in \mathcal{L}(\Sym{4})$ has eigenvalues $\lambda_1 \geq \lambda_2 \geq \cdots \geq \lambda_{10} \geq 0$. Then $\rho$ is absolutely symmetric PPT if and only if the following four matrices are positive semidefinite:
    \begin{align*}
        & \begin{bmatrix}
            2\lambda_{10} & \lambda_3 & \lambda_2 & \lambda_1 \\
            \lambda_3 & 2\lambda_4 & \lambda_5 & \lambda_6 \\
            \lambda_2 & \lambda_5 & 2\lambda_7 & \lambda_8 \\
            \lambda_1 & \lambda_6 & \lambda_8 & 2\lambda_9 \\
        \end{bmatrix}, & &
        \begin{bmatrix}
            2\lambda_{10} & \lambda_3 & \lambda_2 & \lambda_1 \\
            \lambda_3 & 2\lambda_4 & \lambda_5 & \lambda_7 \\
            \lambda_2 & \lambda_5 & 2\lambda_6 & \lambda_8 \\
            \lambda_1 & \lambda_7 & \lambda_8 & 2\lambda_9 \\
        \end{bmatrix}, \\
        & \begin{bmatrix}
            2\lambda_{10} & \lambda_6 & \lambda_2 & \lambda_1 \\
            \lambda_6 & 2\lambda_5 & \lambda_4 & \lambda_3 \\
            \lambda_2 & \lambda_4 & 2\lambda_7 & \lambda_8 \\
            \lambda_1 & \lambda_3 & \lambda_8 & 2\lambda_9 \\
        \end{bmatrix}, & &
        \begin{bmatrix}
            2\lambda_{10} & \lambda_7 & \lambda_2 & \lambda_1 \\
            \lambda_7 & 2\lambda_5 & \lambda_4 & \lambda_3 \\
            \lambda_2 & \lambda_4 & 2\lambda_6 & \lambda_8 \\
            \lambda_1 & \lambda_3 & \lambda_8 & 2\lambda_9 \\
        \end{bmatrix}.
    \end{align*}
\end{conjecture}

Indeed, we have shown numerically that there does not exist a set of just three $4 \times 4$ symmetric matricizations whose positive semidefiniteness implies absolutely symmetric PPT of the associated quantum state, and a few million random examples suggest that the four matrices of Conjecture~\ref{conj:d4_appt} suffice. However, a proof remains elusive (and a proof along the lines of the proof of Corollary~\ref{cor:asymppt_from_eigenvalues_3dim} would be extremely long and tedious).

\section{Spectra of (Decomposable) Symmetric Entanglement Witnesses}\label{sec:witness_spectra}

We now return to the problem of bounding the possible spectra of symmetric entanglement witnesses. We start with the qubit (i.e., $d = 2$) case, where we can \emph{almost} get an exact answer.

\begin{theorem}\label{thm:2x2_eigs}
	Consider the following inequalities involving real numbers $\mu_1 \geq \mu_2 \geq \mu_3$:
	\begin{enumerate}[(a)]
		\item $\mu_2 \geq 0$,
		\item $\mu_3 \geq -\sqrt{\mu_1 \mu_2}$, and
		\item $\mu_3 \geq \begin{cases} -\mu_1/4-\mu_2 & \text{if} \ \ \ \mu_2 < \mu_1/4 \\ -\sqrt{\mu_1\mu_2} & \text{if} \ \ \ \mu_2 \geq \mu_1/4. \end{cases}$
	\end{enumerate}
	If inequalities (a) and (b) hold then there exists a (decomposable) symmetric entanglement witness in $\mathcal{L}(\Sym{2})$ with eigenvalues $\mu_1 \geq \mu_2 \geq \mu_3$. Conversely, the eigenvalues of every decomposable symmetric entanglement witness in $\mathcal{L}(\Sym{2})$ satisfy inequalities (a) and (b), and the eigenvalues of every symmetric entanglement witness in $\mathcal{L}(\Sym{2})$ satisfy inequalities (a) and (c).
\end{theorem}

We conjecture that the eigenvalues of every symmetric entanglement witness in $\mathcal{L}(\Sym{2})$ must satisfy inequalities (a) and (b) (instead of the slightly weaker (a) and (c)), but have been unable to prove it. Inequalities (b) and (c) are the same as each other when $\mu_2 \geq \mu_1/4$, but inequality (c) is strictly weaker than (b) when $\mu_2 < \mu_1/4$, as illustrated in Figure~\ref{fig:2x2_eig_bounds}.

\begin{figure}[!htb]
    \centering
    \begin{tikzpicture}[scale=5]
		\foreach \x in {0.125,0.25,...,1.0} {
			\draw[ultra thin,gray] (2*\x, -1.1) -- (2*\x, 0.05);
			\node[anchor=south] at (2*\x, 0.05) {$\x$};
		}
		\foreach \y in {-0.2,-0.4,-0.6,-0.8,-1.0} {
			\draw[ultra thin,gray] (2*-0.025, \y) -- (2*1.0, \y);
			\node[anchor=east] at (2*-0.025, \y) {$\y$};
		}
		\draw[thick] (2*-0.025, 0) -- coordinate (x axis mid) (2*1.0, 0);
		\node[anchor=east] at (2*-0.025, 0) {$0$};
		\node[above=1cm] at (x axis mid) {$\mu_2$};
		\draw[thick] (2*0, -1.1) -- coordinate (y axis mid) (2*0, 0.05);
		\node[anchor=south] at (2*0, 0.05) {$0$};
		\node[rotate=90, above=1.5cm] at (y axis mid) {Minimal $\mu_3$};
		
		\draw[dashed,samples=26, domain=2*0:2*0.25] plot (\x, {-sqrt(\x/2)});
		\draw[domain=2*0.25:2*1.0, samples=76, /pgf/fpu, /pgf/fpu/output format=fixed] plot (\x,{-sqrt(\x/2)});
		\node[anchor=south west, yshift=-0.05cm] at (2*0.5, {-sqrt(0.5)}) {\begin{tabular}{c} $\mu_3 \geq -\sqrt{\mu_1\mu_2}$ \\ (when $\mu_2 \geq \mu_1/4 = 0.25$) \end{tabular}};
		\draw (2*0,-0.25) -- (2*0.25,-0.5);
		\node[anchor=north east, yshift=1.4cm, xshift=-1.3cm] at (2*0.6, {(-0.25-0.6)}) {\begin{tabular}{c} $\mu_3 \geq -\mu_1/4-\mu_2$ \\ (when $\mu_2 < \mu_1/4 = 0.25$) \end{tabular}};
		\node[circle, fill, inner sep=1.6pt] at (2*0.25, -0.5) {};
		\node[anchor=south west] at (2*0.25, -0.5) {$(0.25, -0.5)$};
	\end{tikzpicture}

    \caption{Lower bound on the least eigenvalue ($\mu_3$) of a symmetric entanglement witness in $\mathcal{L}(\Sym{2})$, as a function of $\mu_2$ (with $\mu_1$ scaled to $1$). The solid line shows inequality~(c) from Theorem~\ref{thm:2x2_eigs}, while the dashed line shows inequality~(b), which we conjecture is tight.}\label{fig:2x2_eig_bounds}
\end{figure}
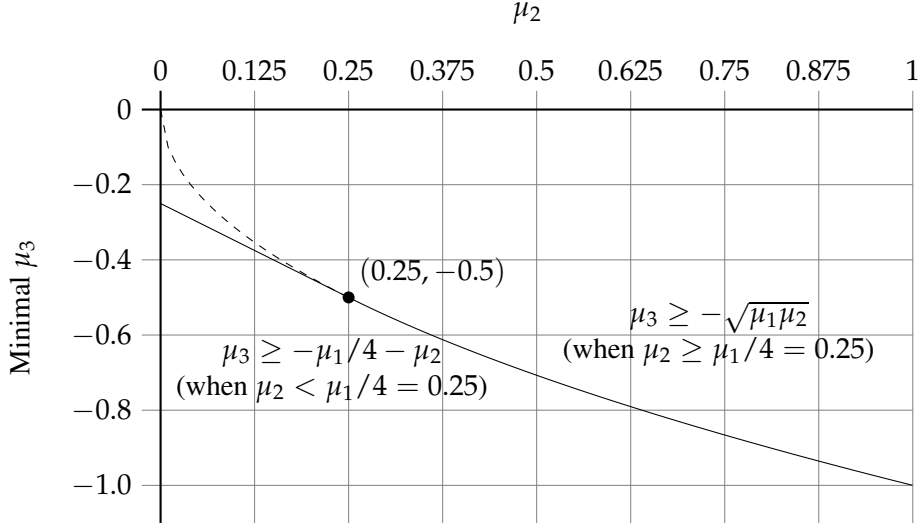

\begin{proof}[Proof of Theorem~\ref{thm:2x2_eigs}.]
	We start by proving that if (a) and (b) hold then we can construct a (decomposable) symmetric entanglement witness $W \in \mathcal{L}(\Sym{2})$ with eigenvalues $\mu_1 \geq \mu_2 \geq \mu_3$. To this end, define
	\begin{align*}
	    X = \begin{bmatrix} \mu_1 & 0 & 0 & \mu_3 \\
		0 & 0 & 0 & 0 \\
		0 & 0 & 0 & 0 \\
		\mu_3 & 0 & 0 & \mu_2 \end{bmatrix},
	\end{align*}
	which is positive semidefinite by inequalities~(a) and~(b)
	(the only part of this claim that is perhaps not immediate is the fact that $\mu_3 \leq \sqrt{\mu_1 \mu_2}$, but this inequality holds due to the fact that $\mu_3 \leq \mu_2 \leq \sqrt{\mu_1 \mu_2}$).
	It is then straightforward to check that $W = P_{\vee} X^\Gamma P_{\vee}$ is a symmetric entanglement witness with eigenvalues $\mu_1$, $\mu_2$, $\mu_3$, and $0$,
	(with corresponding eigenvectors $\mathbf{e_1} \otimes \mathbf{e_1}$, $\mathbf{e_2} \otimes \mathbf{e_2}$, $\mathbf{e_2} \otimes \mathbf{e_1} + \mathbf{e_1} \otimes \mathbf{e_2}$, and $\mathbf{e_2} \otimes \mathbf{e_1} - \mathbf{e_1} \otimes \mathbf{e_2}$, respectively).
	
	Conversely, to see that the eigenvalues of a symmetric entanglement witness $W \in \mathcal{L}(\Sym{2})$ must satisfy inequalities (a) and (c), we first recall from \cite[Corollary~15]{LJ20} that $W$ can have at most $1$ negative eigenvalue, which is equivalent to inequality (a).
	
	To see that inequality~(c) holds, notice that $\tr(W\rho) \geq 0$ holds for all symmetric separable quantum states $\rho \in \mathcal{L}(\Sym{2})$, so it holds for all \emph{absolutely} symmetric separable $\rho$ as well. In particular, this means that
	\begin{align}\label{ineq:2sim_wit_ev_bound}
	    \tr(WU\rho U^*) & \geq 0
	\end{align}
	for all unitaries $U \in \mathcal{L}(\Sym{2})$ and all absolutely symmetric separable $\rho \in \mathcal{L}(\Sym{2})$. By choosing $U$ so that $W$ and $\rho$ are diagonal in the same basis, we see that Inequality~\eqref{ineq:2sim_wit_ev_bound} implies
	\begin{align}\label{ineq:2sim_wit_lammu}
	    \lambda_1\mu_3 + \lambda_2\mu_2 + \lambda_3\mu_1 & \geq 0
	\end{align}
	for all $\lambda_1 \geq \lambda_2 \geq \lambda_3 \geq 0$ that are eigenvalues of an absolutely symmetric separable $\rho \in \mathcal{L}(\Sym{2})$.
	
	By Theorem~\ref{thm:asymppt_from_eigenvalues_2dim}, we know that $\lambda_1 \leq 2\sqrt{\lambda_2\lambda_3}$. In order to determine which values of $\mu_1 \geq \mu_2 \geq \mu_3$ make Inequality~\eqref{ineq:2sim_wit_lammu} hold, we can assume without loss of generality that $\lambda_1 = 2\sqrt{\lambda_2\lambda_3}$ (and hence $2\sqrt{\lambda_2\lambda_3} \geq \lambda_2$, so $\lambda_3 \geq \lambda_2/4$). If we then rescale the $\lambda$s so that $\lambda_2 = 1$, Inequality~\eqref{ineq:2sim_wit_lammu} becomes
	\begin{align}\label{ineq:2sim_wit_lammub}
	    2\sqrt{\lambda_3}\mu_3 + \mu_2 + \lambda_3\mu_1 \geq 0
	\end{align}
	for all real numbers $1/4 \leq \lambda_3 \leq 1$.
	
	If $|\mu_3| \geq \mu_1/2$ then we can choose $\lambda_3 = \mu_3^2/\mu_1^2$, in which case Inequality~\eqref{ineq:2sim_wit_lammub} simplifies to
	\[
	    \mu_3 \geq -\sqrt{\mu_1 \mu_2}.
	\]
	On the other hand, if $|\mu_3| < \mu_1/2$ then we can choose $\lambda_3 = 1/4$, in which case Inequality~\eqref{ineq:2sim_wit_lammub} simplifies to
	\[
	    \mu_3 \geq -(\mu_1/4 + \mu_2),
	\]
	which completes this part of the proof.
	
	All that remains is to show that the eigenvalues of every \emph{decomposable} symmetric entanglement witness satisfy inequality~(b). We leave this to slightly later, right after the proof of Theorem~\ref{thm:dxd_eigenvalues_sew}.
\end{proof}

In higher dimensions, it becomes much trickier to obtain bounds on the possible spectra of symmetric entanglement witnesses. However, we can use semidefinite programming to get bounds on the spectra of \emph{decomposable} symmetric entanglement witnesses, much like semidefinite programming was used to obtain bounds on decomposable (non-symmetric) entanglement witnesses in \cite{JP18}.

Before stating our main result in this area, we need to briefly introduce some new notation. Define function $p^{\uparrow},p^{\downarrow} : \mathbb{R}^n \rightarrow \mathbb{R}^n$ by
\begin{align}\begin{split}\label{eq:defn_p}
    p^{\uparrow}(x_1,x_2,\ldots,x_n) & := \left(\sum_{j=1}^n x_j, \sum_{j=1}^{n-1} x_j, \ldots, \sum_{j=1}^2 x_j, x_1\right), \quad \text{and} \\
    p^{\downarrow}(x_1,x_2,\ldots,x_n) & := \left(\sum_{j=1}^n x_j, \sum_{j=2}^n x_j, \ldots, \sum_{j=n-1}^n x_j, x_n\right),
\end{split}\end{align}
and let $L_k : \mathbb{R}^{d(d+1)/2} \rightarrow \mathcal{L}(\mathbb{R}^d)$ be the linear transformation that sends a vector to its $k$-th symmetric matricization (we do not care about the particular ordering of these matricizations, but just that the $L_1$, $L_2$, $\ldots$, $L_{(d(d+1)/2)!}$ linear transformations capture all possible such matricizations). Up to scaling, the adjoint (dual) transformation of $L_k$ with respect to the standard inner products on $\mathbb{R}^{d(d+1)/2}$ and $\mathcal{L}(\mathbb{R}^d)$ simply reads off the entries from the upper-triangular portion of a matrix in some order. For example, if
\begin{align}\label{eq:d2_special_symmat}
    L_1(x_1, x_2, x_3) = \begin{bmatrix}
        2x_2 & x_1 \\ x_1 & 2x_3
    \end{bmatrix} \quad \text{then} \quad L_1^*\left(\begin{bmatrix}
        y_{1,1} & y_{1,2} \\ y_{1,2} & y_{2,2}
    \end{bmatrix}\right) = 2(y_{1,2},y_{1,1},y_{2,2}).
\end{align}

We then have the following bounds on the eigenvalues of decomposable symmetric entanglement witnesses:

\begin{theorem}\label{thm:dxd_eigenvalues_sew}
    If a decomposable symmetric entanglement witness in $\mathcal{L}(\Sym{d})$ has eigenvalues $\mu_1 \geq \mu_2 \geq \cdots \geq \mu_{d(d+1)/2}$ then there exist positive semidefinite matrices $Y_k \in \mathcal{L}(\mathbb{R}^d)$ for $1 \leq k \leq (d(d+1)/2)!$ such that
    \[
        \sum_k p^{\uparrow}(L_k^*(Y_k)) \leq p^{\downarrow}(\bm{\mu}),
    \]
    where $\bm{\mu} = (\mu_1,\mu_2,\ldots,\mu_{d(d+1)/2})$, $p^{\uparrow}$ and $p^{\downarrow}$ are as in Equation~\eqref{eq:defn_p} (with $n = d(d+1)/2$), and the inequality is meant entrywise.
\end{theorem}

\begin{proof}
    By using the same argument that we used to prove Inequality~\eqref{ineq:2sim_wit_lammu} in the proof of Theorem~\ref{thm:2x2_eigs}, we see that
    \[
        \sum_{j=1}^{d(d+1)/2} \lambda_{n-j+1}\mu_j \geq 0
    \]
    for all $\lambda_1 \geq \lambda_2 \geq \cdots \geq \lambda_{d(d+1)/2} \geq 0$ that are eigenvalues of absolutely symmetric PPT states. It follows that if $\mu_1 \geq \mu_2 \geq \cdots \geq \mu_{d(d+1)/2}$ are the eigenvalues of a symmetric entanglement witness, then the optimal value of the following semidefinite program (in the variable $\bm{\lambda} = (\lambda_1,\lambda_2,\ldots,\lambda_{d(d+1)/2})$; see \cite{Wat18} or \cite[Section~3.C]{JohALA} for an introduction to semidefinite programming) is non-negative:
	\begin{align}\begin{split}\label{sdp:sew_eig_primal}
    	\text{minimize: } & \ \sum_{j=1}^{d(d+1)/2} \lambda_{n-j+1}\mu_j \\
    	\text{subject to: } & \ L_k(\bm{\lambda}) \ \text{is PSD for all $k$} \\
    	& \ \lambda_1 \geq \lambda_2 \geq \cdots \geq \lambda_{d(d+1)/2} \geq 0 \\
    	& \ \sum_{j=1}^{d(d+1)/2} \lambda_j = 1.
	\end{split}\end{align}
	We note that the final constraint $\sum_{j=1}^{d(d+1)/2} \lambda_j = 1$ above is not actually required in this semidefinite program, but it makes it easier to demonstrate that strong duality holds (which we will need to do shortly).
	
	After some simplifying, the dual of this semidefinite program has the following form (in the variables $c \in \mathbb{R}$ and $Y_k \in \mathcal{L}(\mathbb{R}^d)$):
	\begin{align}\begin{split}\label{sdp:sew_eig_dual}
    	\text{maximize: } & \ c \\
    	\text{subject to: } & \ \sum_k p^{\uparrow}(L_k^*(Y_k)) + c\cdot p^{\uparrow}(\bm{1}) \leq p^{\downarrow}(\bm{\mu}) \\
    	& \ Y_k \ \text{is PSD for all $k$},
	\end{split}\end{align}
	where $\bm{1} = (1,1,\ldots,1)$ is the vector whose entries all equal $1$.
	
	All that remains is to show that the primal-dual pair of semidefinite programs~\eqref{sdp:sew_eig_primal} and~\eqref{sdp:sew_eig_dual} satisfy strong duality, so they have the same objective value. Indeed, if we can show that, it would imply that the primal problem~\eqref{sdp:sew_eig_primal} has a non-negative objective value if and only if there is a feasible point of the dual problem~\eqref{sdp:sew_eig_dual} with $c \geq 0$, which is equivalent to existence of PSD matrices $X_k$ such that $\sum_k p^{\uparrow}(L_k^*(Y_k)) \leq p^{\downarrow}(\bm{\mu})$, as desired.
	
	To this end, simply notice that we can find a strictly feasible point of the primal problem~\eqref{sdp:sew_eig_primal} by choosing the $\lambda$'s to be very close (but not equal) to each other. Similarly, we can find a strictly feasible point of the dual problem~\eqref{sdp:sew_eig_dual} by choosing the $X_k$'s to be positive definite and $c$ to be sufficiently large and negative. It follows that the Slater conditions for strong duality \cite[Theorem~1.18]{Wat18} are satisfied, so the semidefinite programs~\eqref{sdp:sew_eig_primal} and~\eqref{sdp:sew_eig_dual} have the same optimal value, so there is a feasible point of the problem~\eqref{sdp:sew_eig_dual} with $c \geq 0$, which completes the proof.
\end{proof}

The criterion of Theorem~\ref{thm:dxd_eigenvalues_sew} can be checked via semidefinite programming, though the size of the semidefinite program grows extremely quickly with the dimension $d$. MATLAB code that implements it using the CVX toolbox \cite{cvx} when $d \leq 5$ (and can thus show that certain sets of eigenvalues can \emph{not} belong the any decomposable symmetric entanglement witness) is available from \cite{Pip21}.

In small dimensions, the criterion of Theorem~\ref{thm:dxd_eigenvalues_sew} simplifies considerably, since you do not need to use every symmetric matricization transformation $L_k$, but only those that are required to show absolute symmetric separability. For example, if $d = 2$ then Theorem~\ref{thm:asymppt_from_eigenvalues_2dim} tells us that we only need to consider the single symmetric matricization $L_1$ from Equation~\eqref{eq:d2_special_symmat}. Theorem~\ref{thm:dxd_eigenvalues_sew} then tells us that if a decomposable symmetric entanglement witness in $\mathcal{L}(\Sym{2})$ has eigenvalues $\mu_1 \geq \mu_2 \geq \mu_3$ then there must exist a positive semidefinite matrix $Y \in \mathcal{L}(\mathbb{R}^2)$ such that
\[
    p^{\uparrow}(L_1^*(Y)) \leq p^{\downarrow}(\bm{\mu}).
\]
If we fill in the details of this vector inequality, we see that this is equivalent to the following set of three scalar inequalities:
\begin{align*}
    y_{2,2} + y_{1,1} + y_{1,2} & \leq \mu_3 + \mu_2 + \mu_1, \\
    y_{1,1} + y_{1,2} & \leq \mu_3 + \mu_2, \\
    y_{1,2} & \leq \mu_3.
\end{align*}

It is a straightforward exercise to show that such a PSD $Y \in \mathcal{L}(\mathbb{R}^2)$ exists if and only if $\mu_3 \geq -\sqrt{\mu_1\mu_2}$ (in which case we can choose $y_{1,2} = \min\{\mu_3,0\}$, $y_{1,1} = \mu_2$, $y_{2,2} = \mu_1$, and otherwise no such PSD $Y$ exists). This is exactly inequality~(b) from Theorem~\ref{thm:2x2_eigs}, which thus completes its proof.

Similarly, if we use the fact from Corollary~\ref{cor:asymppt_from_eigenvalues_3dim} that we only need a single symmetric matricization in the $d = 3$ case, we get the following simplification:
 
\begin{theorem}\label{thm:3x3_eigenvalues_sew}
	If a decomposable symmetric entanglement witness in $\mathcal{L}(\Sym{3})$ has eigenvalues $\mu_1 \geq \mu_2 \geq \cdots \geq \mu_6$ then there exists a positive semidefinite matrix $Y \in \mathcal{L}(\mathbb{R}^3)$ such that the following inequalities all hold:
    \begin{align*}
        y_{1,1} + y_{3,3} + y_{2,3} + y_{2,2} + y_{1,2} + y_{1,3} & \leq \mu_6 + \mu_5 + \mu_4 + \mu_3 + \mu_2 + \mu_1, \\
        y_{3,3} + y_{2,3} + y_{2,2} + y_{1,2} + y_{1,3} & \leq \mu_6 + \mu_5 + \mu_4 + \mu_3 + \mu_2, \\
        y_{2,3} + y_{2,2} + y_{1,2} + y_{1,3} & \leq \mu_6 + \mu_5 + \mu_4 + \mu_3, \\
        y_{2,2} + y_{1,2} + y_{1,3} & \leq \mu_6 + \mu_5 + \mu_4, \\
        y_{1,2} + y_{1,3} & \leq \mu_6 + \mu_5, \\
        y_{1,3} & \leq \mu_6.
    \end{align*}
\end{theorem}

While it is perhaps difficult to use Theorem~\ref{thm:3x3_eigenvalues_sew} analytically, we note again that it is straightforward to check it numerically via semidefinite programming, which we have done in MATLAB at \cite{Pip21} using the CVX toolbox \cite{cvx}.

\subsection{Negative Eigenvalues of Decomposable Symmetric Entanglement Witnesses}

We already mentioned at the end of Section~\ref{sec:decomp_sew} that a symmetric entanglement witness in $\mathcal{L}(\Sym{d})$ has at most $\binom{d}{2} = d(d-1)/2$ negative eigenvalues \cite[Corollary~15]{LJ20}. This bound is known to be tight for symmetric entanglement witnesses, and we now show that it is even tight for symmetric entanglement witnesses of the form $P_{\vee} X^\Gamma P_{\vee}$, where $X \in \mathcal{L}(\mathbb{C}^d \otimes \mathbb{C}^d)$ is positive semidefinite.

\begin{theorem}\label{thm:dsew_max_neg_evals}
    For every integer $d \geq 1$, there exists a positive semidefinite matrix $X \in \mathcal{L}(\mathbb{C}^d \otimes \mathbb{C}^d)$ such that $P_{\vee} X^\Gamma P_{\vee}$ has exactly $d(d-1)/2$ strictly negative eigenvalues.
\end{theorem}

Before we can prove this theorem, we first need to prove the following helper lemma, which we believe is of independent interest.

\begin{lemma}\label{lem:npt_sym_sub}
    There exists a subspace $\mathcal{S} \subseteq \Sym{d}$ of dimension $d(d-1)/2$ such that every quantum state $\rho \in \mathcal{L}(\Sym{d})$ with range contained in $\mathcal{S}$ is such that $\rho^\Gamma$ has a negative eigenvalue. In particular,
    \[
        \mathcal{S} := \mathrm{span}\big\{ (\mathbf{e_i} \otimes \mathbf{e_{i+k}} - \mathbf{e_{i+1}} \otimes \mathbf{e_{i+k+1}}) + (\mathbf{e_{i+k}} \otimes \mathbf{e_i} - \mathbf{e_{i+k+1}} \otimes \mathbf{e_{i+1}}) : 1 \leq i < d, k \geq 0, i+k < d \big\}
    \]
    is such a subspace.
\end{lemma}

Since the proof of the above lemma is quite long, and also quite similar to the proof of the analogous non-symmetric result (\cite[Theorem~1]{Joh13B}), we defer it to the appendix. Instead, we now show how to use this lemma to prove Theorem~\ref{thm:dsew_max_neg_evals}.

\begin{proof}[Proof of Theorem~\ref{thm:dsew_max_neg_evals}.]
    To construct a decomposable symmetric entanglement witness $W \in \mathcal{L}(\Sym{d})$
    with the desired $d(d-1)/2$ negative eigenvalues, let $P$ be the orthogonal projection onto the subspace $\mathcal{S} \subseteq \Sym{d}$ described by Lemma~\ref{lem:npt_sym_sub}.
    
    By that lemma, if $\rho \in \mathcal{L}(\Sym{d})$ is such that $\rho^\Gamma$ is positive semidefinite (for convenience, we will call such $\rho$ ``symmetric PPT'' throughout this proof), then $P \rho P \neq \rho$, so $\Tr(P\rho) < 1$. Since the set of symmetric PPT $\rho$ is compact, there exists a real constant $0 < c < 1$ such that $\Tr(P\rho) \leq c$ for all symmetric PPT $\rho$.
    
    If we define the operator $W = P_{\vee} - \frac{1}{c}P$ then it is easily verified that $W$ has $d(d-1)/2$ negative eigenvalues and $\Tr(W \rho) \geq 0$ for all symmetric PPT $\rho$. Standard duality techniques then imply that we can write $W = P_{\vee} X^\Gamma P_{\vee} + P_{\vee}YP_{\vee}$ for some positive semidefinite $X, Y \in \mathcal{L}(\mathbb{C}^d \otimes \mathbb{C}^d)$. Since $P_{\vee}YP_{\vee}$ is positive semidefinite and $W$ has $d(d-1)/2$ negative eigenvalues, $P_{\vee} X^\Gamma P_{\vee}$ must have at least $d(d-1)/2$ negative eigenvalues. On the other hand, $P_{\vee} X^\Gamma P_{\vee}$ is a symmetric entanglement witness, so it has at most $d(d-1)/2$ negative eigenvalues, which completes the proof.
\end{proof}

The operators described by Theorem~\ref{thm:dsew_max_neg_evals} form a subset of symmetric entanglement witnesses, but not decomposable ones, since it does not require $X = P_{\vee}X P_{\vee}$. However, numerics suggest that if we add this extra restriction in, then we can still get $d(d-1)/2$ negative eigenvalues. For example, if $d = 4$ then the positive semidefinite matrix
\[
    X = \left[\begin{array}{cccc|cccc|cccc|cccc}
         3 & \cdot & -1 & \cdot & \cdot & -2 & \cdot & -1 & -1 & \cdot & \cdot & \cdot & \cdot & -1 & \cdot & \cdot \\
         \cdot & 7 & \cdot & \cdot & 7 & \cdot & 1 & \cdot & \cdot & 1 & \cdot & 3 & \cdot & \cdot & 3 & \cdot \\
         -1 & \cdot & 7 & \cdot & \cdot & 5 & \cdot & \cdot & 7 & \cdot & 1 & \cdot & \cdot & \cdot & \cdot & -1 \\
         \cdot & \cdot & \cdot & 9 & \cdot & \cdot & 6 & \cdot & \cdot & 6 & \cdot & \cdot & 9 & \cdot & \cdot & \cdot \\\hline
         \cdot & 7 & \cdot & \cdot & 7 & \cdot & 1 & \cdot & \cdot & 1 & \cdot & 3 & \cdot & \cdot & 3 & \cdot \\
         -2 & \cdot & 5 & \cdot & \cdot & 6 & \cdot & 1 & 5 & \cdot & \cdot & \cdot & \cdot & 1 & \cdot & \cdot \\
         \cdot & 1 & \cdot & 6 & 1 & \cdot & 7 & \cdot & \cdot & 7 & \cdot & 1 & 6 & \cdot & 1 & \cdot \\
         -1 & \cdot & \cdot & \cdot & \cdot & 1 & \cdot & 7 & \cdot & \cdot & 5 & \cdot & \cdot & 7 & \cdot & -1 \\\hline
         -1 & \cdot & 7 & \cdot & \cdot & 5 & \cdot & \cdot & 7 & \cdot & 1 & \cdot & \cdot & \cdot & \cdot & -1 \\
         \cdot & 1 & \cdot & 6 & 1 & \cdot & 7 & \cdot & \cdot & 7 & \cdot & 1 & 6 & \cdot & 1 & \cdot \\
         \cdot & \cdot & 1 & \cdot & \cdot & \cdot & \cdot & 5 & 1 & \cdot & 6 & \cdot & \cdot & 5 & \cdot & -2 \\
         \cdot & 3 & \cdot & \cdot & 3 & \cdot & 1 & \cdot & \cdot & 1 & \cdot & 7 & \cdot & \cdot & 7 & \cdot \\\hline
         \cdot & \cdot & \cdot & 9 & \cdot & \cdot & 6 & \cdot & \cdot & 6 & \cdot & \cdot & 9 & \cdot & \cdot & \cdot \\
         -1 & \cdot & \cdot & \cdot & \cdot & 1 & \cdot & 7 & \cdot & \cdot & 5 & \cdot & \cdot & 7 & \cdot & -1 \\
         \cdot & 3 & \cdot & \cdot & 3 & \cdot & 1 & \cdot & \cdot & 1 & \cdot & 7 & \cdot & \cdot & 7 & \cdot \\
         \cdot & \cdot & -1 & \cdot & \cdot & \cdot & \cdot & -1 & -1 & \cdot & -2 & \cdot & \cdot & -1 & \cdot & 3
    \end{array}\right] = P_{\vee}XP_{\vee}
\]
is such that $P_{\vee} X^\Gamma P_{\vee}$ has $d(d-1)/2 = 6$ negative eigenvalues (beating the $5$ negative eigenvalues that we were able to find when $X$ was chosen to be rank-$1$ back in Figure~\ref{fig:dimensional_eigenvalue_comparisons}). MATLAB code that similarly finds decomposable entanglement witnesses in $\mathcal{L}(\Sym{d})$ with $d(d-1)/2$ negative eigenvalues (for values of $d \leq 12$ or so) via semidefinite programming and the CVX and QETLAB toolboxes \cite{cvx,qetlab} is provided at \cite{Pip21}.

\section{Conclusions and Open Problems}\label{sec:conclusions}

In this work, we explored two spectral problems that arise from entanglement in symmetric quantum systems: we established bounds on the possible spectra of symmetric entanglement witnesses, and we also established bounds on the possible spectra of states that remain separable under any symmetric (but global) change of basis. There are numerous remaining open problems and avenues for future research extending from this paper:

\begin{enumerate}[(a)]
    \item What can be said about the spectra of \emph{anti}symmetric (i.e., fermionic) entanglement witnesses? We did not consider this problem in this paper since the techniques that we used are not particularly applicable in the fermionic setting where a ``separable'' pure state is one of the form $\mathbf{v} \wedge \mathbf{w} = \mathbf{v} \otimes \mathbf{w} - \mathbf{w} \otimes \mathbf{v}$. It is known that antisymmetric entanglement witnesses in $\mathcal{L}(\mathbb{C}^d \wedge \mathbb{C}^d)$ have at most $\binom{d-2}{2} = (d-2)(d-3)/2$ negative eigenvalues \cite[Corollary~16]{LJ20}, but it is not clear how to obtain non-trivial bounds on those negative eigenvalues.
    
    \item Similarly, what can be said about the spectra of $k$-entanglement witnesses (i.e., witnesses of Schmidt number $k$) \cite{TH00}? This problem seems to be significantly more challenging when $k \geq 2$ than when $k = 1$ for two reasons: the smallest non-trivial case involves matrices of size at least $9 \times 9$, and we cannot make use of the transpose map in this higher-rank setting (since it is not $2$-positive).
    
    \item Can we strengthen Theorem~\ref{thm:2x2_eigs} into an exact characterization of the eigenvalues of a symmetric entanglement witness in $\mathcal{L}(\Sym{2})$? We conjecture that inequalities (a) and (b) from that theorem provide the correct characterization.
    
    \item Can we find an explicit characterization of the minimal set of positive semidefinite matrices that need to be checked in Theorem~\ref{thm:asymppt_from_eigenvalues} in order to ensure that a state is absolute symmetric PPT? We showed that checking positive semidefiniteness of $(d(d+1)/2)!$ different $d \times d$ matrices suffices when the local dimension is $d$, but we expect that far fewer matrices actually suffice. We showed in Theorem~\ref{thm:asymppt_from_eigenvalues_2dim} and Corollary~\ref{cor:asymppt_from_eigenvalues_3dim} that if $d = 2$ or $d = 3$ then it suffices to check positive semidefiniteness of just one $d \times d$ matrix, and we conjectured in Conjecture~\ref{conj:d4_appt} that if $d = 4$ then it suffices to check four $d \times d$ matrices.

    \item The set of absolutely symmetric separable states is trivially contained within the set of absolutely symmetric PPT states. However, it is not clear whether or not these sets are actually equal to each other. We do not know if this question will be easier or harder to solve than the corresponding question in the non-symmetric setting, which is also open \cite{AJR15}.
    
    \item Can we strengthen Theorem~\ref{thm:dsew_max_neg_evals} to require $X = P_{\vee} X P_{\vee}$, so that $P_{\vee} X^\Gamma P_{\vee}$ is a decomposable entanglement witness? One way to prove this would be to strengthen Lemma~\ref{lem:npt_sym_sub} to show that $P_{\vee}\rho^\Gamma P_{\vee}$, not just $\rho^\Gamma$, has a negative eigenvalue.
\end{enumerate}


\noindent \textbf{Acknowledgements.} N.J.\ was supported by NSERC Discovery Grant number RGPIN-2016-04003.

\bibliographystyle{alpha}
\bibliography{bib}

\section*{Appendix: Proof of Lemma~\ref{lem:npt_sym_sub}}
\addcontentsline{toc}{section}{Appendix: Proof of Lemma~\ref{lem:npt_sym_sub}}

We largely follow the proof of \cite[Theorem~1]{Joh13B}. Let $\mathcal{S} \subseteq \Sym{d}$ be the subspace described by the statement of the lemma (which clearly has the correct dimension $\binom{d}{2} = d(d-1)/2$), and let $\Delta : \Sym{d} \rightarrow \mathcal{L}(\mathbb{R}^d)$ be the linear transformation with the property that $\Delta(\mathbf{v} \otimes \mathbf{v}) = \mathbf{v}\mathbf{v}^T$. It is then straightforward to see that $\Delta(\mathcal{S})$ is the space of $d \times d$ symmetric (not Hermitian!) matrices whose forward diagonals all sum to $0$.

Now suppose $\rho \in \mathcal{L}(\Sym{d})$ is a quantum state with range contained in $\mathcal{S}$. Then we can write $\rho = \sum_j \mathbf{v_j}\mathbf{v}_{\mathbf{j}}^*$ for some $\{\mathbf{v_j}\} \subset \mathcal{S}$. We define $W := \rho^\Gamma$, $X^{(j)} := \Delta(\mathbf{v_j})$, we use $w_{a,m;b,n}$ to denote the $(b,n)$-entry of the $(a,m)$-block of $W$ (i.e., $w_{a,m;b,n} = (\mathbf{e_a} \otimes \mathbf{e_b})^*W(\mathbf{e_m} \otimes \mathbf{e_n})$), and we similarly use $x_{a,m}^{(j)}$ to denote the $(a,m)$-entry of $X^{(j)}$ (i.e., $x_{a,m}^{(j)} = \mathbf{e}_{\mathbf{a}}^*X^{(j)}\mathbf{e_m}$). Since $X^{(j)}$ is symmetric, we have $x_{a,m}^{(j)} = x_{m,a}^{(j)}$ for all $a$ and $m$. Our goal is to show that $W$ has a negative eigenvalue, which we do by finding a $2 \times 2$ principal submatrix of it with negative determinant.

Straightforward calculation shows that
\[
	w_{a,m;b,n} = \sum_j x^{(j)}_{a,n} \overline{x^{(j)}_{b,m}},
\]
so the $2 \times 2$ principal submatrix of $W$ corresponding to rows and columns $\mathbf{e_{a}}\otimes\mathbf{e_{n}}$ and $\mathbf{e_{b}}\otimes\mathbf{e_{m}}$ (yes, the indices here are ``swapped'' on purpose) has the form
\[
    Y^{(a,b;m,n)} := \begin{bmatrix}
        w_{a,a;n,n}   &   w_{a,b;n,m} \\
	    w_{b,a;m,n}   &   w_{b,b;m,m}
	\end{bmatrix} = \begin{bmatrix}
        \sum_j \big|x^{(j)}_{a,b}\big|^2   &  \sum_j x^{(j)}_{a,m} \overline{x^{(j)}_{b,n}} \\
	    \sum_j x^{(j)}_{b,n} \overline{x^{(j)}_{a,m}}   &   \sum_j\big|x^{(j)}_{b,m}\big|^2
	\end{bmatrix}.
\]

To find a submatrix of this form that has negative determinant, we recall that each $X^{(j)}$ matrix has its forward diagonals summing to $0$. There thus exists a farthest-to-the-top-right forward diagonal that is not identically $0$. We claim that there exists some choice of $(a,m) \neq (n,b)$ (with $a < n$ and $m < b$ without loss of generality) such that the $(a,m)$- and $(n,b)$-entries are both on this diagonal (which guarantees that $x_{a,b}^{(j)} = 0$ for all $j$, since the $(a,b)$-entry is farther to the top-right), and furthermore $\sum_j x^{(j)}_{a,m} \overline{x^{(j)}_{b,n}} \neq 0$. If we can find such $a$, $b$, $m$, and $n$, then we will be done, since we would then have
\[
    \det(Y^{(a,b;m,n)}) = -\left|\sum_j x^{(j)}_{a,m} \overline{x^{(j)}_{b,n}}\right|^2 < 0.
\]

To prove our claim that we can choose $a$, $b$, $m$, and $n$ so that $\sum_j x^{(j)}_{a,m} \overline{x^{(j)}_{b,n}} \neq 0$, suppose for the sake of establishing a contradiction that $\sum_j x^{(j)}_{a,m} \overline{x^{(j)}_{b,n}} = 0$ for all $(a,m) \neq (n,b)$ with the $(a,m)$- and $(n,b)$-entries on the same forward diagonal. Let $L$ be the number of entries in that forward diagonal, and let $J$ be the number of terms in the sums over $j$ that we have been working with. Arrange the entries $x^{(j)}_{a,m}$ of that forward diagonal into a $J \times L$ matrix $Z$, whose $(j,\ell)$-entry is $x^{(j)}_{a,m}$, where the $(a,m)$ indexes the $\ell$-th entry of the forward diagonal.

Then the condition $\sum_j x^{(j)}_{a,m} \overline{x^{(j)}_{b,n}} = 0$ for all $(a,m) \neq (n,b)$ is equivalent to the columns of $Z$ being mutually orthogonal, while the forward diagonals of $X$ summing to $0$ is equivalent to the rows of $Z$ summing to $0$. However, if the rows of $Z$ sum to $0$, then the columns of $Z$ sum to the zero vector, so they are linearly dependent and cannot possibly be mutually orthogonal. This is the desired contradiction that completes the proof.\hfill $\square$
\end{document}